\documentclass{article}
\usepackage{adjustbox}
\usepackage{algorithm}
\usepackage{algorithmic}
\usepackage{amsmath, amsfonts}
\usepackage{amsthm}
\usepackage{bibentry}
\usepackage{bm}
\usepackage{booktabs}
\usepackage[small]{caption}
\usepackage{enumitem}
\usepackage{graphicx}
\usepackage[hidelinks]{hyperref}
\usepackage{cleveref}
\usepackage{ijcai22}
\usepackage[utf8]{inputenc}
\usepackage{multirow}
\usepackage{natbib}
\usepackage{paralist}
\usepackage{soul}
\usepackage{subcaption}
\usepackage{thmtools}
\usepackage{thm-restate}
\usepackage{times}
\usepackage{url}
\hypersetup{draft}

\newtheorem{theorem}{Theorem}
\newtheorem{corollary}{Corollary}
\newtheorem{lemma}{Lemma}
\newtheorem{observation}{Observation}

\def\sp {{\bm\sigma}}
\def\ssp {{\bm\sigma}_{ij}}
\def\sspp {{\bm\sigma}_{ik}}
\def\u {{\sf U}}

\def\hpoa {{\sf PoA}}
\def\hpos {{\sf PoS}}
\def\ut {{\sf SUM}}
\def\hn {{\sf DoI}}
\def\i {1_{ij}({\bm\sigma})}
\def\k {1_{ik}({\bm\sigma})}

\newcommand{\ra}[1]{\renewcommand{\arraystretch}{#1}}

\title{Tolerance is Necessary for Stability: Single-Peaked Swap Schelling Games}

\author{
	Davide Bilò$^1$
	\and
	Vittorio Bilò$^2$\and
	Pascal Lenzner$^{3}$\And
	Louise Molitor$^3$
	\affiliations
	$^1$University of L'Aquila, Italy, davide.bilo@univaq.it \\
	$^2$University of Salento, Italy, vittorio.bilo@unisalento.it\\
	$^3$Hasso Plattner Institute, University of Potsdam, Germany, \{firstname.lastname\}@hpi.de
}

\begin{document}
	
	\maketitle
	
	\begin{abstract}
		Residential segregation in metropolitan areas is a phenomenon that can be observed all over the world. Recently, this was investigated via game-theoretic models. There, selfish agents of two types are equipped with a monotone utility function that ensures higher utility if an agent has more same-type neighbors. The agents strategically choose their location on a given graph that serves as residential area to maximize their utility. However, sociological polls suggest that real-world agents are actually favoring mixed-type neighborhoods, and hence should be modeled via non-monotone utility functions. 
		To address this, we study Swap Schelling Games with single-peaked utility functions.
		Our main finding is that tolerance, i.e., agents favoring fifty-fifty neighborhoods or being in the minority, is necessary for equilibrium existence on almost regular or bipartite graphs.  
		Regarding the quality of equilibria, we derive (almost) tight bounds on the Price of Anarchy and the Price of Stability. In particular, we show that the latter is constant on bipartite and almost regular graphs.    
	\end{abstract}
	
	\section{Introduction}
		Residential segregation is defined as the physical separation of two or more groups into different neighborhoods~\citep{MD88}. It is pervasive in metropolitan areas, where large homogeneous regions inhabited by residents belonging to the same ethnic group emerged over time\footnote{See the racial dot map~\citep{C13} at \url{https://demographics.coopercenter.org/racial-dot-map/} for examples from the US.}. 
		
		For more than five decades, residential segregation has been intensively studied by sociologists, as a high degree of segregation has severe consequences for the inhabitants of homogeneous neighborhoods. It negatively impacts their health~\citep{acevedo2003residential}, their mortality~\citep{jackson2000relation}, and in general their socioeconomic conditions~\citep{massey1993american}. While in the early days of research on segregation the emergence of homogeneous neighborhoods was attributed to the individual intolerance of the citizens, it was shown by~\citet{Schelling71} that residential segregation also emerges in a tolerant population. In his landmark model, he considers two types of agents that live on a line or a grid as residential area. Every agent has a tolerance level $\tau$ and is content with her position, if at least a $\tau$-fraction of her direct neighbors are of her type. Discontent agents randomly jump to other empty positions or swap positions with another discontent agent. Schelling found that even for $\tau < \frac{1}{2}$, i.e., even if everyone is content with being in the minority within her neighborhood, random initial placements are over time transformed to placements having large homogeneous regions, i.e., many agents are surrounded by same-type neighbors, by the individual random movements of the agents. It is important to note that the agent behavior is driven by a slight bias towards preferring a certain number of same-type neighbors and that this bias on the microlevel is enough to tip the macrolevel state towards segregation. Schelling coined the term ``micromotives versus macrobehavior'' for such phenomena~\citep{S06}.    
		
		Since its inception, Schelling's influential model was thoroughly studied by sociologists, mathematicians and physicists via computer simulations.
		But only in the last decade progress has been made to understand the involved random process from a theoretical point of view. Even more recently, the Algorithmic Game Theory and the AI communities became interested in residential segregation and game-theoretic variants of Schelling's model were studied~\citep{CLM18,E+19,BBLM20,AEGISV21,KKV21,BSV21}. In these strategic games the agents do not perform random moves but rather jump or swap to positions that maximize their utility. These models incorporate utility functions that are monotone in the fraction of same-type neighbors, i.e., the utility of an agent is proportional to the fraction of same-type neighbors in her neighborhood. See Figure~\ref{fig:utility}~(left).
		\begin{figure}
			\centering 
			\includegraphics[width=\linewidth]{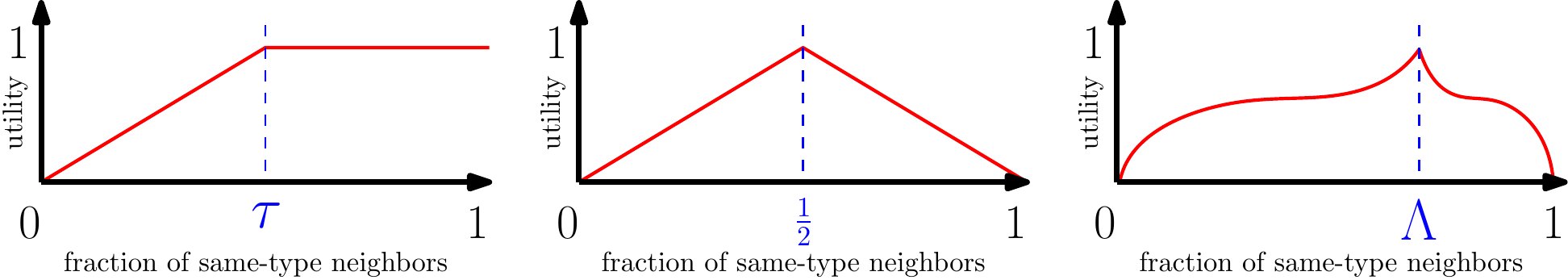}
			\caption{Left: example of the monotone utility functions employed in recent related work. Middle and right: example of a single-peaked utility function considered in this paper.}
			\label{fig:utility}
		\end{figure}
		However, representative sociological polls, in particular data from the General Social Survey\footnote{Since 50 years the GSS is regularly conducted in the US and it is a valuable and widely used data set for social scientists. } (GSS)~\citep{gss},  indicate that this assumption of monotone utility functions should be challenged. For example, in 1982 all black respondents where asked \textit{``If you could find the housing that you would want and like, would you rather live in a neighborhood that is all black; mostly black; half black, half white; or mostly white?''} and $54\%$ responded with \textit{``half black, half white''} while only $14\%$ chose \textit{``all black''}. Later, starting from 1988 until 2018 all respondents (of whom on average $78\%$ were white) where asked what they think of \textit{``Living in a neighborhood where half of your neighbors were blacks?''} a clear majority\footnote{In numbers: 1988: $57\%$, 1998: $70\%$, 2008: $79\%$, 2018: $82\%$. In 2018 $33\%$ answered with ``favor'' or ``strongly favor''.} responded \textit{``strongly favor''}, \textit{``favor''} or \textit{``neither favor nor oppose''}. This shows that the maximum utility should not be attained in a homogeneous neighborhood. 
	
		Based on these real-world empirical observations, this paper sets out to explore a game-theoretic variant of Schelling's model with non-monotone utility functions. In particular, we will focus on single-peaked utility functions with maximum utility at a $\Lambda$-fraction of same-type neighbors (see Figure~\ref{fig:utility} (middle and right)), with $\Lambda \in (0,1)$, satisfying mild assumptions. More precisely, we only require a function $p(x)$ to be zero-valued at $x=0,1$, to be strictly increasing in the interval $[0,\Lambda]$ and to be such that $p(x)=p(\frac{\Lambda(1-x)}{1-\Lambda})$ for each $x\in [\Lambda,1]$, that is, both sides of $p$ approach the peak, one from the left and the other from the right, in the same way, up to a rescaling due to the width of their domains ($[0,\Lambda]$, vs. $[\Lambda,1]$). Our main findings shed light on the existence of equilibrium states and their quality in terms of the recently defined degree of integration~\citep{AEGISV21} that measures the number of agents that live in a heterogeneous neighborhood.
	
		\paragraph{Model}
			We consider a strategic game played on a given underlying connected graph $G=(V,E)$, with $|V| = n$ and $|E| = m$.

			For any node $v \in V$, let the \emph{closed neighborhood} of $v$ in~$G$ be $$N[v]=\{v\}\cup\{u\in V:\{v,u\}\in E\}$$ where $\delta(v)=|N[v]|-1$ denotes the \emph{degree of~$v$}, and $\delta(G)$ and $\Delta(G)$ denote the minimum and the maximum degree over all nodes in $G$, respectively. We call a graph $G$ $\delta$-\emph{regular}, if $\delta(G) = \Delta(G) = \delta$, and \emph{almost regular}, if $\Delta(G)-\delta(G) \leq 1$. 
			We denote with $\alpha(G)$ the \emph{independence number} of $G$, i.e., the cardinality of the maximum independent set in $G$.
	
			A \emph{Single-Peaked Swap Schelling Game $(G,b,\Lambda)$}, called \emph{the game}, is defined by a graph $G = (V,E)$, a positive integer $b\leq n/2$ and a peak position $\Lambda$. There are $n$ strategic agents who choose nodes in~$V$ such that every node is occupied by exactly one agent. Every agent belongs to one of two types that are associated with the colors blue and red. There are $b$ blue agents and $r = n-b$ red agents, with blue being the color of the minority type. Let $c(i)$ be the color of agent~$i$.
	
			A \emph{strategy profile $\sp$} is an $n$-dimensional vector where all strategies are pairwise disjoint, i.e., $\sp$ is a permutation of $V$. The $i$-th entry of $\sp$ corresponds to the strategy of the $i$-th agent. We treat $\sp$ as a bijective function mapping agents to nodes, with $\sp^{-1}$ being its inverse function. Thus, any strategy profile $\sp$ corresponds to a bi-coloring of $G$ in which exactly~$b$ nodes of $G$ are colored blue and $n-b$ are colored red. We say that agent $i$ \emph{occupies node $v$ in~$\sp$} if the $i$-th entry of~$\sp$, denoted as $\sp(i)$, is $v$ and, equivalently, if $\sp^{-1}(v) = i$. We use the notation $\i$, with $\i = 1$ if agents~$i$ and $j$ occupy adjacent nodes in~$\sp$ and $\i = 0$ otherwise. When $\i=1$, we say that agents are {\em adjacent}.  
	
			For an agent $i$ and a feasible strategy profile $\sp$, we denote the set of nodes of $G$ which are occupied by agents having the same color as agent $i$ by $$C_i(\sp)=\{v\in V:c(\sp^{-1}(v))=c(i)\}.$$ Observe that $C_i(\sp)$ includes node $\sp(i)$. Let $$f_i(\sp):=\frac{|N[{\sp(i)}]\cap C_i(\sp)|}{|N[{\sp(i)}]|}$$ be the fraction of agents of her own color in $i$'s neighborhood including herself. Thus, agents are aware of their own contribution to the diversity of their neighborhood The utility of an agent~$i$ in $\sp$ is defined as $\u_i(\sp)=p(f_i(\sp))$, where~$p$ is a single-peaked function with domain $[0,1]$ and peak at $\Lambda \in (0,1)$ that satisfies the following two properties: 
			\begin{compactitem}
				\item[(i)] $p$ is a strictly monotonically increasing function in the interval $[0,\Lambda]$ with $p(0)=0$;
				\item[(ii)] for each $x\in [\Lambda,1]$,  $p(x)=p(\frac{\Lambda(1-x)}{1-\Lambda})$ and $p(\Lambda)=1$.
			\end{compactitem}
			Each agent aims at maximizing her utility. We say an agent~$i$ is {\em below the peak} when $f_i(\sp)<\Lambda$, {\em above the peak} when $f_i(\sp)>\Lambda$, at {\em the peak} when $f_i(\sp)=\Lambda$, and {\em segregated} when $f_i(\sp)=1$.	
			A game $(G,b,\Lambda)$ depends also on the choice of $p$. However, as all our results are independent of~$p$, we remove it from the notation for the sake of simplicity.
	
			An agent can change her strategy only via a \emph{swap}, i.e., exchanging node occupation with another agent. Consider two agents $i$ and $j$, on nodes $\sp(i)$ and $\sp(j)$, respectively, performing a swap. This yields the new strategy profile $\ssp$. As agents are strategic, we only consider {\em profitable swaps}, i.e., swaps which strictly increase the utility of both agents. Hence, profitable swaps can only occur between agents of different colors. A strategy profile $\sp$ is a \emph{swap equilibrium }(SE), if $\sp$ does not admit profitable swaps, i.e., if for each pair of agents $i,j$, we have $\u_i(\sp)\geq\u_i(\ssp)$ or $\u_j(\sp)\geq\u_j(\ssp)$. 
	
			We measure the quality of a strategy profile~$\sp$ via the \emph{degree of integration} ($\hn$), defined by the number of non-segregated agents. The $\hn$ is a simple segregation measure that captures how many agents have contact with other-type agents. We prefer it to the standard utilitarian welfare since it measures segregation independently of the value of $\Lambda$.
			For any fixed game $(G,b,\Lambda)$ let~$\sp^*$ denote a feasible strategy profile maximizing the $\hn$ and let~${\sf SE}(G,b,\Lambda)$ denote the set of swap equilibria for $(G,b,\Lambda)$.  
			We study the impact of the agents' selfishness by analyzing the \emph{Price of Anarchy} (PoA), which is defined as $${\sf PoA}(G,b,\Lambda) = \frac{\hn(\sp^*)}{\min_{\sp \in {\sf SE}(G,b,\Lambda)}\hn(\sp)}$$ and the \emph{Price of Stability} (PoS), which is defined as $${\sf PoS}(G,b,\Lambda) = \frac{\hn(\sp^*)}{\max_{\sp \in {\sf SE}(G,b,\Lambda)}\hn(\sp)}.$$
	
			We investigate the \emph{finite improvement property (FIP)} \citep{MS96}, i.e., if every sequence of profitable swaps is finite, which is equivalent to the existence of an ordinal potential function.	For this, let $$\Phi(\sp) = \left|\left\{\{u,v\}\in E : c(\sp^{-1}(u)) = c(\sp^{-1}(v))\right\}\right|,$$ counting the number of \emph{monochromatic edges} of $G$ under $\sp$, i.e., the edges whose endpoints are occupied by agents of the same color, the \emph{potential function of $\sp$}.
	
		\paragraph{Related Work}
			In the last decade progress has been made to thoroughly understand the involved random process in Schelling's influential model, e.g., \citep{BIK12,BEL14,BIK17}. 
	
			\cite{Zha04,Zha04b} investigated the random Schelling process via evolutionary game theory. In particular, \cite{Zha04b} proposes a model that is similar to our model. There, agents on a toroidal grid graph with degree 4 also have a non-monotone single-peaked utility function. However, in contrast to our model, random noise is added to the utilities and transferable utilities are assumed. Zhang analyzes the Markov process of random swaps and shows that this process converges with high probability to segregated states.  
	
			The investigation of game-theoretic models for residential segregation was initiated by~\cite{CLM18}. There, agents are equipped with a utility function as shown in Figure~\ref{fig:utility} (left) and the finite improvement property and the PoA in terms of the number of content agents is studied. Later, \cite{E+19} significantly extended these results and generalized them to games with more than two agent types. 
	
			\cite{AEGISV21} introduce a simplified model with $\tau=1$. They prove results on the existence of equilibria, in particular that equilbria are not guaranteed to exist on trees, and on the complexity of deciding equilibrium existence. Moreover, they study the PoA in terms of the utilitarian social welfare and in terms of the newly introduced degree of integration, that counts the number of non-segregated agents. For the latter, they give a tight bound of $\frac{n}{2}$ on the PoA and the PoS that is achieved on a tree. In contrast, they derive a constant PoS on paths. \cite{BBLM20} strengthened the PoA results for the simplified swap version w.r.t. the utilitarian social welfare function and investigated the model on almost regular graphs, grids and paths. Additionally, they introduce a variant with locality. The complexity results were extended by~\cite{KBFN21}. \cite{CIT20} and \cite{KKV21a} considered generalized variants.   
	
			Recently, a model was introduced where the agent itself is included in the computation of the fraction of same-type neighbors~\citep{KKV21}. We adopt this modified version also in our model. \cite{BSV21} consider the number of agents with non-zero utility as social welfare function. They prove hardness results for computing the social optimal state and they discuss other stability notations such as Pareto optimality.    
	
			Also related are hedonic games~\citep{DG80,BJ02} where selfish agents form coalitions and the utility of an agent only depends on her coalition. Especially close are hedonic diversity games~\citep{BEI19,BE20}, where agents of different types form coalitions and the utility depends also on the type distribution in a coalition. 
	
			Our main focus is on single-peaked utility functions. This can be understood as single-peaked preferences, which date back to~\cite{Black48} and are a common theme in the Economics and Game Theory literature. In particular, such preferences yield favorable behavior in the above mentioned hedonic diversity games and in the realm of voting and social choice~\citep{Walsh07,YCE13,BSU13,EFS14,BBHH15}.
	
		\paragraph{Our Contribution}
			\begin{table*}[h]
				\centering	
				\scriptsize
				\ra{1.3}
				\begin{adjustbox}{max width=\textwidth}
					\begin{tabular}{@{}lrlrl}
						\toprule
						graph classes & \multicolumn{2}{c}{Equilibrium Existence} & \multicolumn{2}{c}{Finite Improvement Property} \\
						\cmidrule(l{2em}r{1em}){2-3} \cmidrule(l{0.5em}r{0.5em}){4-5} 
						\textbf{arbitrary} & $\times$ (Thm.~\ref{thm:noSE}) & $\Lambda > 1/2$ & $\times$ (Thm.~\ref{thm:noSE}, \ref{thm:theoremirc}) & $\Lambda \geq 1/2$ \\
						& \checkmark (Thm.~\ref{thm:bipartite}) & $\frac{1}{\delta(G)+1}\leq\Lambda\leq 1/2$, $\alpha(G)+1\geq b$ \\
						\textbf{bipartite} & \checkmark (Cor.~\ref{cor:bipartite}) & $\Lambda = 1/2$ \\
						\textbf{$1$-regular} & \checkmark (Thm.~\ref{thm:conv_regular}) & $\Lambda \leq 1/2$ & \checkmark (Thm.~\ref{thm:conv_regular}) & $\Lambda \leq 1/2$ \\ 
						\textbf{$2$-regular} & $\times$ (Thm.~\ref{thm:noSE}) & $\Lambda > 1/2$ & $\times$ (Thm.~\ref{thm:noSE}) & $\Lambda > 1/2$ \\
						\midrule
						& \multicolumn{2}{c}{Price of Anarchy} & \multicolumn{2}{c}{Price of Stability} \\
						\cmidrule(l{2em}r{1em}){2-3} \cmidrule(l{0.5em}r{0.5em}){4-5}   
						\textbf{arbitrary} & $\leq  \min\bigl\{\Delta(G),\frac{n}{b+1}, \frac{(\Delta+1)b}{b+1}\bigr\}$ (Thm.~\ref{poa-gen}) & & $\geq \Omega\left(\sqrt{n\Lambda}\right)$ (Thm.~\ref{posLB_general}) \\
						\textbf{bipartite} & $\geq \frac{n-1}{3}$ (Thm.~\ref{poa-gen}) & $b = 1$ & $2$ (Thm.~\ref{thm:pos_bipartite}, \ref{thm:pos_bipartite_lb}) & $\Lambda = 1/2$ \\
						& $\geq \frac{n}{b+1}$ (Thm.~\ref{poa-gen}) & $b > 1$ & \\
						\textbf{regular} & $\leq \min\bigl\{(\delta+1)/2, n/2b\bigr\}$ (Thm.~\ref{ub-reg}) & $\Lambda<1/\delta$ \\
						& $\geq \frac{\delta+1}{2}-\frac{\delta+1}{4\delta+2}$ (Thm.~\ref{lb-reg}) & $\Lambda \leq 1/2$, $\delta \geq 2$ \\
						\textbf{1-regular} &&&$1$ (Thm.~\ref{thm:PoS_almost_cubic_graphs}), \ref{thm:pos_for_regular_graphs_and_large_b})& $\Lambda \leq 1/2$, $\Delta(G) \leq 3$ or \\
						&&&& $\Lambda \in \left[\frac{1}{\delta(G)+1}, 1/2\right]$, $b \geq \alpha(G)$ \\
						&&& $\min\bigl\{\Delta(G)+1,\text{O}(1/\Lambda)\bigr\}$ (Thm.~\ref{thm:pos_for_regular_graphs_and_small_b}) & $\Lambda \leq 1/2$, $b < \alpha(G)$ \\
						&&& $\text{O}(1)$  (Cor.~\ref{cor:pos_almost_regular}) & $\Lambda \leq 1/2$ \\
						\textbf{ring} & $>2 -\epsilon$ (Thm.~\ref{thm:theorempoarings}) \\
						& $>3/2 -\epsilon$ (Thm.~\ref{thm:theorempoarings}) & $\Lambda < 1/2$ \\	
						\bottomrule
					\end{tabular}
				\end{adjustbox}
				\vspace*{+0.1cm}
				\caption{Result overview. We investigate the existence of equilibria, the finite improvement property, the PoA and the PoS. The ``\checkmark'' symbol denotes that the respective property holds, the ``$\times$'' means the opposite. The respective conditions are stated next to the result. $\epsilon$ is a constant larger than zero. ``$1$-regular'' stands for almost regular graphs. Note, PoS results for almost regular graphs hold for regular graphs as well. For the PoA the stated lower bounds of other graph classes hold for arbitrary graphs as well.}\label{tbl:previous_results}
			\end{table*}
	
			In this work we initiate the study~of game-theoretic models for residential segregation with non-monotone utility functions. This departs from the recent line of work focusing on monotone utility functions and it opens up a promising research direction. Non-monotone utility functions are well-justified by real-world data and hence might be more suitable for modeling real-world segregation. 
	
			We focus on a broad class of non-monotone utility functions well-known in Economics and Algorithmic Game Theory: single-peaked utilities. We emphasize that our results hold for all such functions that satisfy our mild assumptions. See Table~\ref{tbl:previous_results} for a detailed result overview.
	
			For games with integration-oriented agents, i.e., $\Lambda \leq 1/2$, we show that swap equilibria exist on almost regular graphs and that improving response dynamics are guaranteed to converge to such stable states. Moreover, for $\Lambda = \frac{1}{2}$ swap equilibria exist on the broad class of graphs that admit an independent set that is large enough to accommodate the minority type agents. In particular, this implies equilibrium existence and efficient computability on bipartite graphs, including trees, which is in contrast to the non-existence results by~\cite{AEGISV21}. 
	
			Another contrast are our bounds on the PoA. On general graphs we prove a tight bound on the PoA that depends on~$b$, the number of agents of the minority color, and we give a bound of $\Delta(G)$ on all graphs~$G$, that is asymptotically tight on $\delta$-regular graphs. Also for the PoS we get stronger positive results compared to~\citep{AEGISV21}. For $\Lambda = \frac12$ we give a tight PoS bound of $2$ on bipartite graphs and show that the PoS is $1$ on almost regular graphs with maximum degree $3$, or if the size of the maximum independent set of the graph is at most $b$. The latter implies a PoS of $1$ on regular graphs for balanced games, i.e., if there are equally many agents of both colors. Even more general, for constant $\Lambda \leq \frac12$ we prove a constant PoS on almost regular graphs via a sophisticated proof technique that relies on the greedy algorithm for the \textsc{k-Max-Cut} problem.  
	
	\section{Preliminaries}\label{sec:prelimins}
		In this section, we provide some facts and lemmas that will be widely exploited throughout the paper. We start by observing the following fundamental relationship occurring between $f_i(\sp)$ and $f_j(\ssp)$ for two swapping agents~$i$ and~$j$: 
		\begin{equation}\label{eq1}
			\textrm{if } f_i(\sp)=\frac{x}{y},\footnote{For the sake of conciseness, from now on, whenever we write $f_i(\sp)=x/y$ for some agent $i$, we implicitly mean that $x:=|N[{\sp(i)}]\cap C_i(\sp)|$ and $y:=|N[{\sp(i)}]|$. Observe that, under
			this assumption, $f_i(\sp)=3/6$ is different than $f_i(\sp)=1/2$.}\textrm{ then } f_j(\ssp)=\frac{y+1-x-\i}{y}.
		\end{equation}
		Using property (1), we claim the following observation.
		\begin{restatable}{observation}{observationone}\label{ob1}
			If $f_i(\sp)=x/y < 1/2$, then $f_j(\ssp)>1/2$. If $f_i(\sp)=x/y > 1/2$, then $f_j(\ssp)\leq 1/2$ unless $y=2x-1$ and $\i=0$, for which $f_j(\ssp)=f_i(\sp)=x/y > 1/2$.
		\end{restatable}

		\begin{proof}
			If $f_i(\sp)=x/y < 1/2$, we have $$f_j(\ssp)=\frac{y+1-x-\i}{y}=1-\frac{x}{y}+\frac{1-\i}{y}>\frac 1 2.$$ If $f_i(\sp)=x/y>1/2$, we distinguish among different cases. If $\i=1$, we get $$f_j(\ssp)=\frac{y+1-x-\i}{y}=1-\frac x y<\frac 1 2,$$ while, if $y<2x-1$, which implies $x\geq\frac{y+2}{2}$, it follows 
			\begin{eqnarray*}
				f_j(\ssp)&=&\frac{y+1-x-\i}{y} \\
				&=&1-\frac x y+\frac{1-\i}{y}\\
				&\leq&\frac 1 2-\frac{\i}{y}\leq\frac{1}{2}. 
			\end{eqnarray*}
			Finally, for $y=2x-1$ and $\i=0$, we get $$f_j(\ssp)=\frac{y+1-x-\i}{y}=\frac x y > \frac 1 2.$$
		\end{proof}
	
		\noindent The following series of lemmas characterizes the conditions under which a profitable swap can take place.
	
		\begin{restatable}{lemma}{lemmaone}\label{lemma0} 
			For any $\Lambda\leq 1/2$, no profitable swaps can occur between agents below the peak.
		\end{restatable}

		\begin{proof}
		Fix a strategy profile $\sp$ and two agents~$i$ and $j$, below the peak, who can perform a profitable swap in $\sp$.
		By Observation \ref{ob1}, both $i$ and $j$ are above the peak in $\ssp$. Assume, w.l.o.g., that $f_i(\sp)=x/y<\Lambda$ which yields $$f_j(\ssp)=\frac{y+1-x-\i}{y}>\Lambda.$$ We claim that $\u_j(\ssp)\leq\u_i(\sp)$. By the definition of~$p$, this holds whenever 
		\begin{eqnarray*}
			\frac{x}{y}&\geq&\frac{\Lambda}{1-\Lambda}\left(1-\frac{y+1-x-\i}{y}\right)\\
			&=&\frac{\Lambda}{1-\Lambda}\left(\frac{x}{y}-\frac{1-\i}{y}\right)
		\end{eqnarray*}
		which holds true as $\frac{\Lambda}{1-\Lambda}\leq 1$ and $1-\i\geq 0$.
		
		By applying the same argument, with $i$ and $j$ swapped, we also get $\u_i(\ssp)\leq\u_j(\sp)$. As the swap is profitable, we have $\u_i(\sp)<\u_i(\ssp)$ and $\u_j(\sp)<\u_j(\ssp)$. Putting all these inequalities together, we conclude that $$\u_j(\ssp)\leq\u_i(\sp)<\u_i(\ssp)\leq \u_j(\sp)<\u_j(\ssp),$$ which yields a contradiction.
		\end{proof}
	
		\begin{restatable}{lemma}{lemmatwo}\label{lemma0.1} 
			For any $\Lambda\leq 1/2$, no profitable swaps can occur between adjacent agents at different sides of the peak.
		\end{restatable}
		\begin{proof}
			Assume towards a contradiction, that $i$ and $j$ can perform a profitable swap in $\sp$, and, w.l.o.g., that $f_i({\bm\sigma})=x/y<\Lambda$ and $f_j({\bm\sigma})=x'/y'>\Lambda$.
			By Observation~\ref{ob1}, $j$ ends up above the peak in $\ssp$. As $j$ improves after the swap, we have $${\sf U}_j(\ssp)=p(1-x/y)>{\sf U}_j(\sp)=p(x'/y')$$ which, given that $1-x/y>\Lambda$ and $x'/y'>\Lambda$, yields $$1-x/y<x'/y'.$$ This implies that $$f_i(\ssp)=1-x'/y'<1-1+x/y=x/y=f_i(\sp)$$ which, given that $f_i(\sp)<\Lambda$, contradicts the fact that $i$ improves after the swap.
		\end{proof}
	
		\noindent In the following we present a technical result that will be helpful in proving Lemma~\ref{almost}.
		
		\begin{lemma}\label{lemma4} 
			For any $\Lambda\leq 1/2$, any profitable swap occurring between two agents $i$ and $j$ in a strategy profile $\sp$, with $f_i(\sp)< \Lambda$ and $f_j(\sp)> \Lambda$, requires $\delta(\sp(i))>\delta(\sp(j))$.
		\end{lemma}
		\begin{proof}
			Assume towards a contradiction, that $$\delta(\sp(i))\leq\delta(\sp(j))$$ and $i$ and $j$ can perform a profitable swap in $\sp$, and, w.l.o.g., that $f_i({\bm\sigma})=x/y<\Lambda$ and $f_j({\bm\sigma})=x'/y'>\Lambda$. 
			
			By Lemma \ref{lemma0.1}, it must be $\i=0$.
			By Observation~\ref{ob1}, $j$ ends up above the peak in $\ssp$. As $j$ improves after the swap, we have $${\sf U}_j(\ssp)=p(1-x/y+1/y)>{\sf U}_j(\sp)=p(x'/y')$$ which, given that $1-x/y+1/y>\Lambda$ and $x'/y'>\Lambda$, yields $$1-x/y+1/y<x'/y'.$$ This implies that 
			\begin{eqnarray*}
				f_i(\ssp)&=&1-x'/y'+1/y'\\
				&<&1-1+x/y-1/y+1/y'\\
				&=&x/y+1/y'-1/y.
			\end{eqnarray*}
			Now, as the hypothesis $\delta(\sp(i))\leq\delta(\sp(j))$ can be restated as $y\leq y'$, we derive $$f_i(\ssp)<x/y+1/y'-1/y\leq x/y=f_i(\sp),$$ which, given that $f_i(\sp)<\Lambda$, contradicts the fact that $i$ improves after the swap.
		\end{proof}
	
		\begin{restatable}{lemma}{corollaryone}\label{almost}
			For any $\Lambda\leq 1/2$, no profitable swaps can occur between agents at different sides of the peak in games on almost regular graphs.
		\end{restatable}
	
		\begin{proof}
			Fix a strategy profile $\sp$ and two agents~$i$ and~$j$ at different sides of the peak admitting a profitable swap in~$\sp$. As the game is played on an almost regular graph, by Lemma~\ref{lemma4}, it must be $f_i(\sp)=x/(y+t)<\Lambda$, $f_j(\sp)=x'/y>\Lambda$ and $t\in\{0,1\}$. Moreover, by Lemma \ref{lemma0.1}, we have $\i=0$.
		
			Since $j$ improves after the swap, we have $${\sf U}_j(\ssp)=p((y+t-x+1)/(y+t))>{\sf U}_j(\sp)=p(x'/y)$$ which, given that $(y+t-x+1)/(y+t)>\Lambda$ and $x'/y>\Lambda$, yields $$x'/y>(y+t-x+1)/(y+t).$$ We derive $x'(y+t)>y(y+t-x+1)$, which, given that both sides of the inequality are integers, yields
			\begin{equation}\label{equa1}
				x'(y+t)\geq y(y+t-x+1)+1.
			\end{equation}
			Since $i$ improves after the swap, we have $${\sf U}_i(\ssp)=p((y-x'+1)/y)>{\sf U}_i(\sp)=p(x/(y+t)).$$ We now distinguish between two possible cases: 
			
			{\em (i)} $(y-x'+1)/y\leq\Lambda$ and 
			
			{\em (ii)} $(y-x'+1)/y>\Lambda$.
		
			If case {\em (i)} occurs, it must be $$x/(y+t)<(y-x'+1)/y$$ which is equivalent to $$x'(y+t)<(y+t)(y+1)-xy.$$ Together with inequality (\ref{equa1}), this yields $$y(y+t-x+1)+1<(y+t)(y+1)-xy$$ which is satisfied if and only if $t>1$. Given that $t\in\{0,1\}$, we derive a contradiction.
			
			If case {\em (ii)} occurs, from $(y-x'+1)/y>\Lambda$, we get
			\begin{equation}\label{equa2}
				x'<y(1-\Lambda)+1.
			\end{equation} 
			Since ${\sf U}_i(\ssp)=p((y-x'+1)/y)>{\sf U}_i(\sp)=p(x/(y+t))$, $x/(y+t)<\Lambda$ and $(y-x'+1)/y>\Lambda$, by the definition of~$p$, we derive $$x/(y+t)<\frac{\Lambda}{1-\Lambda}(1-(y-x'+1)/y)=\frac{\Lambda}{1-\Lambda}(x'-1)/y,$$ by which we get $x<\frac{\Lambda(x'-1)(y+t)}{(1-\Lambda)y}$. Together with inequality~(\ref{equa2}), this yields 
			\begin{equation}\label{equa3}
				x<\Lambda(y+t).
			\end{equation}
			By summing up inequalities (\ref{equa2}) and (\ref{equa3}), we get $$x+x'<y+1+\Lambda t.$$ As $x$, $x'$, $y$ are integers, $t\in\{0,1\}$ and $\Lambda\in [0,1/2]$, we derive 
			\begin{equation}\label{equa4}
				x+x'\leq y+1.
			\end{equation}
			Starting from inequality (\ref{equa1}) and then using inequality (\ref{equa4}), we derive
			\begin{eqnarray*}
				x't & \geq & y(y+t-x-x'+1)+1\\
				& \geq & yt+1,
			\end{eqnarray*}
			which, given that $x'\leq y$, is never satisfied when $t\in\{0,1\}$. Thus, also in this case, we derive a contradiction. 
		\end{proof}
		
	\section{Existence of Equilibria}\label{sec-existence}
	
		In this section, we provide existential results for games played on some specific graph topologies. We start by showing that games on almost regular graphs enjoy the FIP property and converge to a SE in at most $m$ swaps in any game in which the peak does not exceed $1/2$. This result does not hold when the peak exceeds $1/2$, as we prove the existence of a game played on a $2$-regular graph (i.e., a ring) admitting no SE.
		
		\begin{theorem}\label{thm:conv_regular}
			For any $\Lambda\leq 1/2$, fix a game $(G,b,\Lambda)$ on an almost regular graph $G$ and a strategy profile~$\sp$. Any sequence of profitable swaps starting from $\sp$ ends in a SE after at most~$m$ swaps.
		\end{theorem}
		\begin{proof}
			We show that, after a profitable swap, $\Phi$ decreases by at least $1$. Consider a profitable swap performed by agents~$i$ and $j$ such that $f_i(\sp)=x/y$ and $f_j(\sp)=x'/(y+t)$, with $t\in\{0,1\}$ since $G$ is almost regular. 
			By Lemmas \ref{lemma0} and \ref{almost}, we have that both $i$ and $j$ are above the peak, i.e., $x/y>\Lambda$ and $x'/(y+t)>\Lambda$. 
			Thus, a necessary condition for the swap to be profitable is that $f_i(\ssp)<f_i(\sp)$ and $f_j(\ssp)<f_j(\sp)$. By Observation \ref{ob1}, the latter yields $$x'/(y+t)>1-x/y+(1-\i)/y,$$ which gives 
			\begin{eqnarray*}
				x'&>&y-x+1-\i+t(1-x/y+(1-\i)/y) \\
				&\geq& y-x+1-\i.
			\end{eqnarray*}
			Since $x$, $x'$, $y$ and $\i$ are integers, we derive $$x'\geq y-x+2-\i.$$ As it holds that $\Phi(\sp)-\Phi(\ssp)$ equals 
			\begin{eqnarray*}
				&~& x-1+x'-1-\\
				&~& \left(y-x-\i+y+t-x'-\i\right)\\
				&=& 2(x+x'-1+\i)-2y-t,
			\end{eqnarray*}
			we get $\Phi(\sp)-\Phi(\ssp) \geq 1.$
		\end{proof}
		
		\begin{restatable}{theorem}{theoremnoSE} \label{thm:noSE}
			For any $\Lambda>1/2$, there exists a game played on a $2$-regular graph admitting no SE.
		\end{restatable}
		
		\begin{proof}
			\noindent Consider an instance of a game played on a ring with~$6$ nodes, where $b=r=3$. Only the following two complementary cases may occur:
			
			Either, the blue agents occupy nodes that induce a path of length 2. In this case, there are two segregated agents of different colors, both with utility $0$. As $p(0)=0$ and $p(x)> 0$ for $x \in (0,1)$, the two agents swap their positions.
			
			Or, there are two neighboring agents $i$ and $j$ of different colors being below the peak. In this case, as $p(1/3)<p(2/3)$, both $i$ and $j$ prefer to swap their positions.
		\end{proof}
		
		\noindent A fundamental question is whether a SE always exists in games with tolerant agents, i.e., for $\Lambda\leq 1/2$. Next result shows that Theorem \ref{thm:conv_regular} cannot be generalized to all graphs.
		
		\begin{restatable}{theorem}{theoremirc}
			\label{thm:theoremirc}
			There cannot exist an ordinal potential function in games on arbitrary graphs for $\Lambda=1/2$.
		\end{restatable}
	
		\begin{proof}
			We prove the statement by providing an improving response cycle where in every step a profitable swap is possible. The construction and arising strategy profiles are shown in Figure~\ref{IRC}. For the sake of simplicity we assume $p(x) = x$ for $x\in [0,1/2]$. However, our result is independent of the choice of $p$. Remember that, since $\Lambda = \frac12$, $p(x) = p(1-x)$ for $x\in [1/2,1]$.
			
			In the initial placement, cf.~Figure~\ref{initialplacement}, agents $a$ and $b$ can swap. By swapping their positions, agent $a$ can increase her utility from $1-\frac57 = \frac27$ to $1-\frac{9}{13} = \frac{4}{13}$ and agent $b$ increases her utility from $\frac{5}{13}$ to $\frac37$. 
			
			Next, agents $c$ and $d$ can swap, cf.~Figure~\ref{firstswap}. Swapping with agent $d$ increases agent $c$'s utility from $1-\frac47 = \frac37$ to~$\frac{6}{13}$, and agent $d$ can increase her utility from $1-\frac{8}{13} = \frac{5}{13}$ to~$1-\frac47 = \frac37$. 
			
			After this, cf.~Figure~\ref{secondswap} agents~$e$ and $f$, and agents $g$ and $h$, respectively, have the opportunity to swap and increase their utility. Agent $e$ increases her utility from $\frac37$ to $\frac12$ while agent~$f$ increases her utility from $1-\frac22 = 0$ to $1-\frac57=\frac27$. Agent $g$ improves her utility from $1-\frac22 = 0$ to $1-\frac35=\frac25$ and agent~$h$ increases her utility from $1-\frac35=\frac25$ to $\frac12$. 
			
			Next, cf.~Figure~\ref{thirdswap}, swaps between the agents $i$ and $j$, and~$k$ and $l$, respectively, are possible. Agent~$i$ can increase her utility from $1-\frac22 = 0$ to $1-\frac45=\frac15$ and~$j$ can increase her utility from $\frac25$ to $\frac12$. Agent $k$ improves her utility from $1-\frac22 = 0$ to $1-\frac35 = \frac25$ and agent $l$ increases her utility from $1-\frac35 = \frac25$ to $\frac12$. 
			
			In the next step, agents $e$ and~$m$ can swap. Agent $e$ increases her utility from $1-\frac22=0$ to $1-\frac35 = \frac25$ and~$m$ increases her utility from $1-\frac35 = \frac25$ to~$\frac12$. 
			
			In a final step agents $n$ and $o$ swap. Swapping with agent $o$ increases agent $n$'s utility from $1-\frac22 = 0$ to $1-\frac{8}{13} = \frac{5}{13}$ and agent $o$ increases her utility from $\frac{6}{13}$ to $\frac12$. The now reached placement, cf.~Figure~\ref{lastswap} is equivalent to the initial placement, cf.~Figure~\ref{initialplacement}. 
			
			However, note that although convergence is not guaranteed there still exists a stable state, cf. Figure~\ref{NE}.	
			\begin{figure}
				\begin{subfigure}{0.4\textwidth}
					\centering
					\includegraphics[scale=0.7]{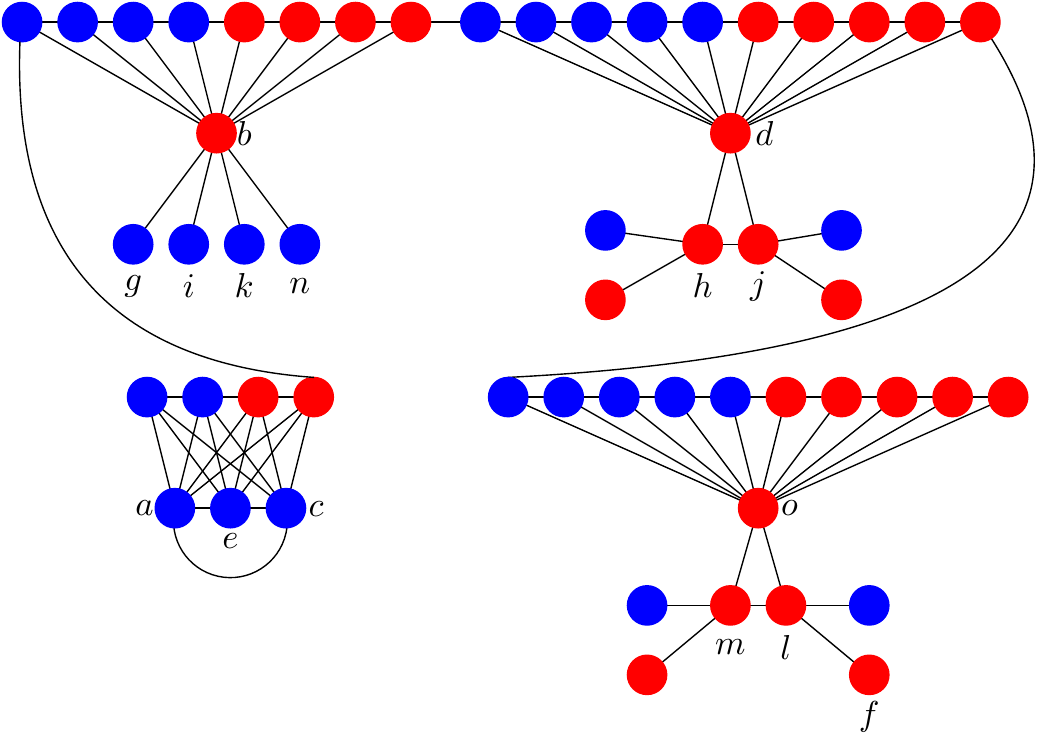}
					\caption{Initial placement}\label{initialplacement}
				\end{subfigure}	
				\begin{subfigure}{0.4\textwidth}
					\centering
					\includegraphics[scale=0.7]{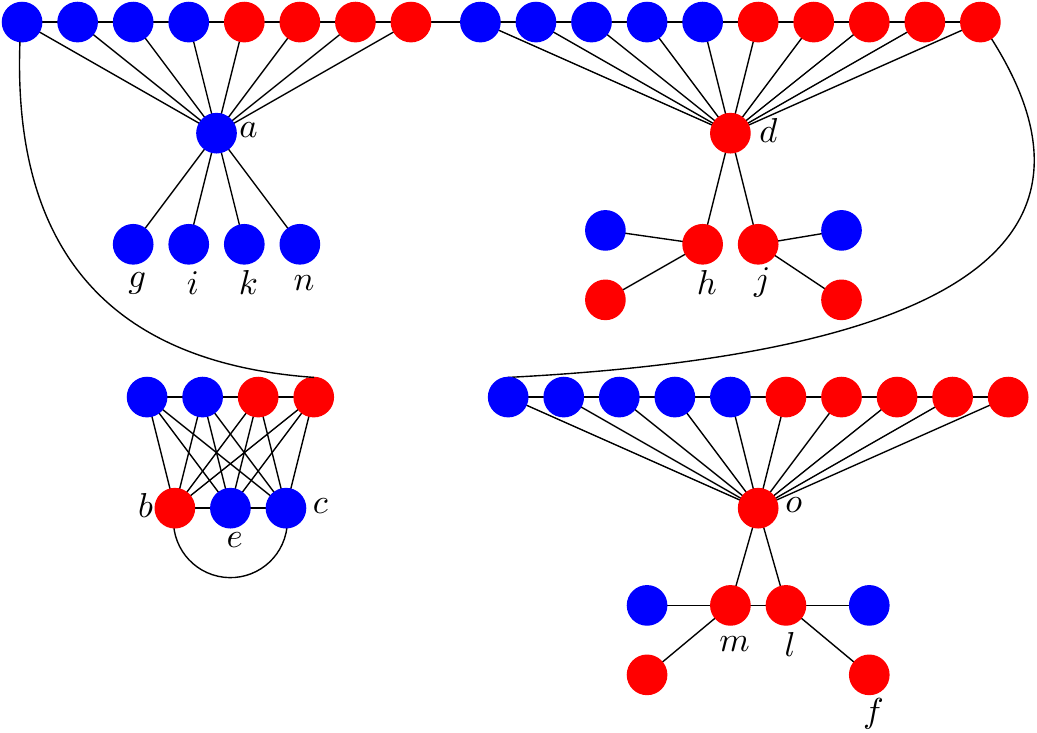}
					\caption{Placement after the first swap of $a$ and $b$}\label{firstswap}
				\end{subfigure}	
				\begin{subfigure}{0.4\textwidth}
					\centering
					\includegraphics[scale=0.7]{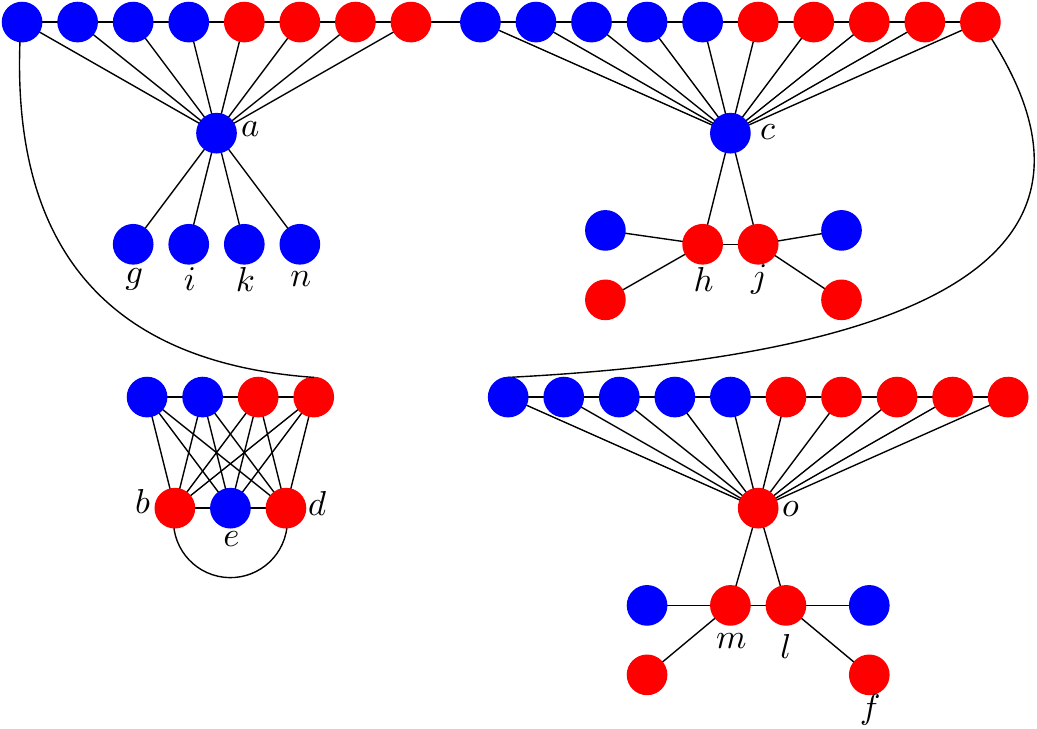}
					\caption{Placement after the second swap of $c$ and $d$}\label{secondswap}
				\end{subfigure}	
				\begin{subfigure}{0.4\textwidth}
					\centering
					\includegraphics[scale=0.7]{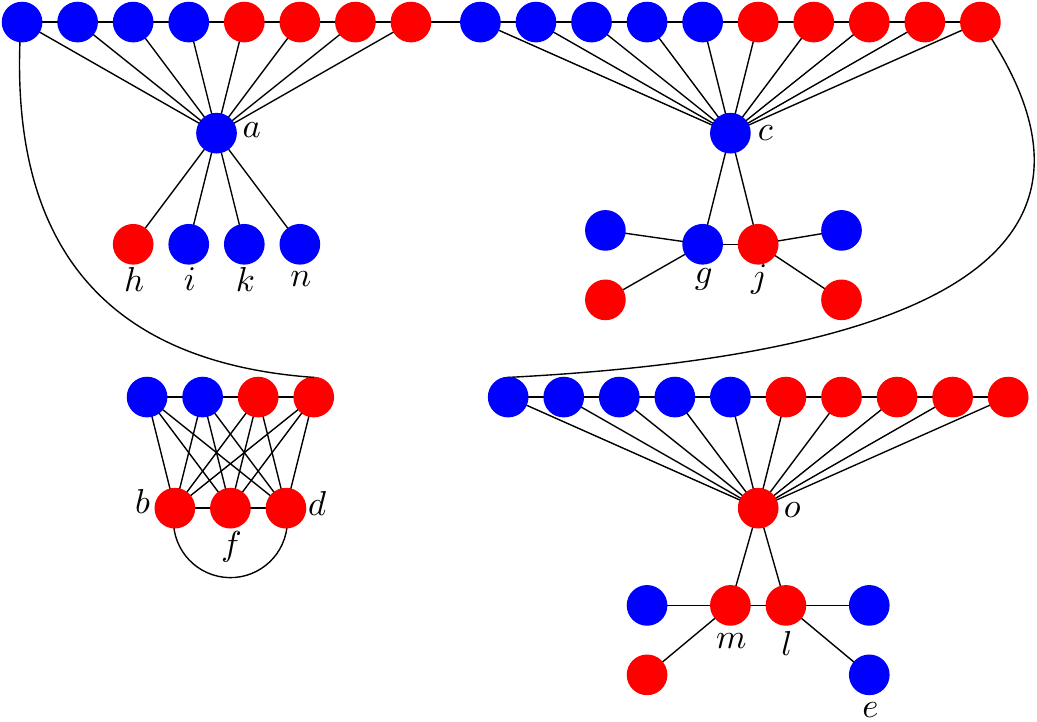}
					\caption{Placement after the third and forth swap of $e$ and $f$, and $g$ and $h$, respectively.}\label{thirdswap}
				\end{subfigure}	
			\end{figure}
			\begin{figure}\ContinuedFloat
				\begin{subfigure}{0.4\textwidth}
					\centering
					\includegraphics[scale=0.7]{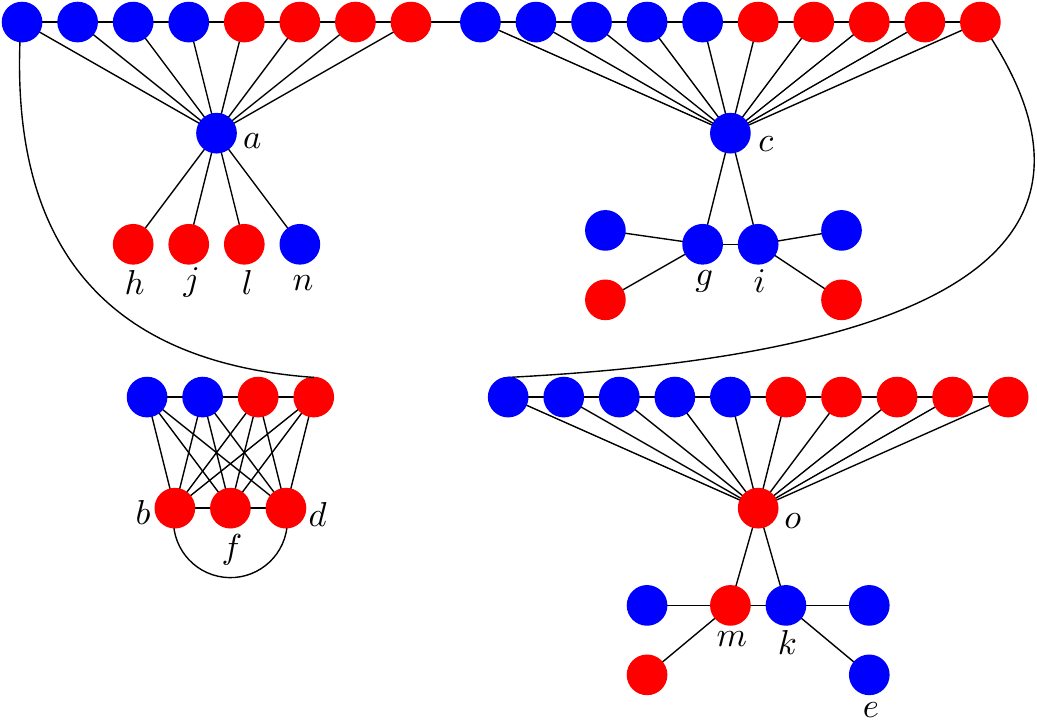}
					\caption{Placement after the fifth and sixth swap of $i$ and $j$, and $k$ and $l$, respectively.}\label{forthswap}
				\end{subfigure}	
				\begin{subfigure}{0.4\textwidth}
					\centering
					\includegraphics[scale=0.7]{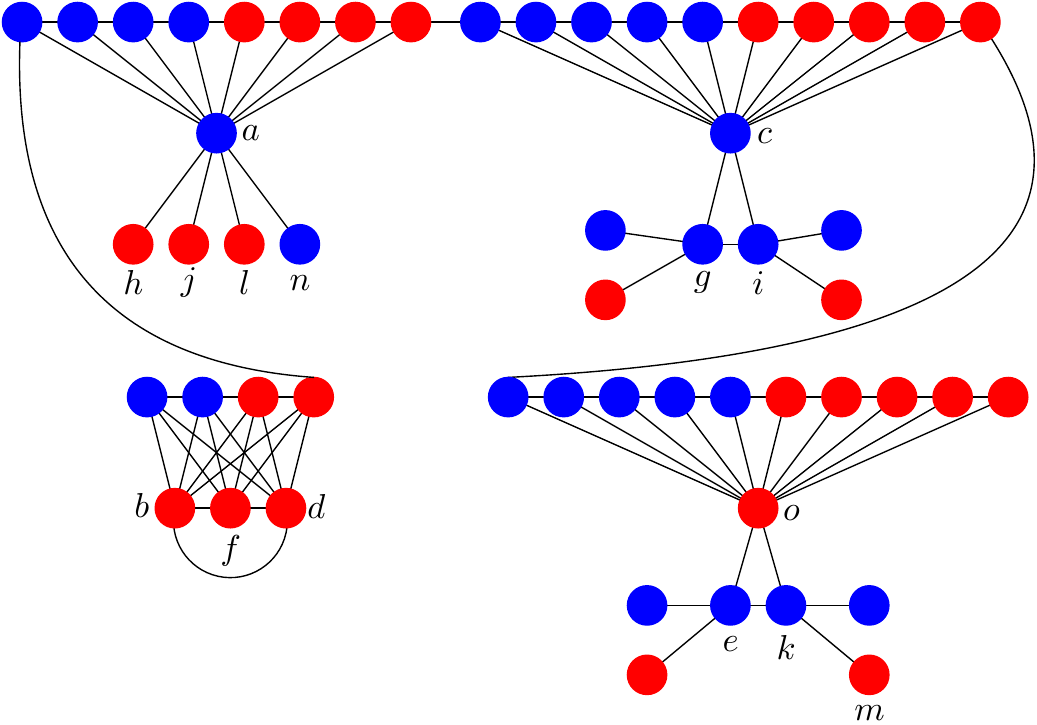}
					\caption{Placement after the seventh swap of $e$ and $m$.}\label{fifthswap}
				\end{subfigure}	
				\begin{subfigure}{0.4\textwidth}
					\centering
					\includegraphics[scale=0.7]{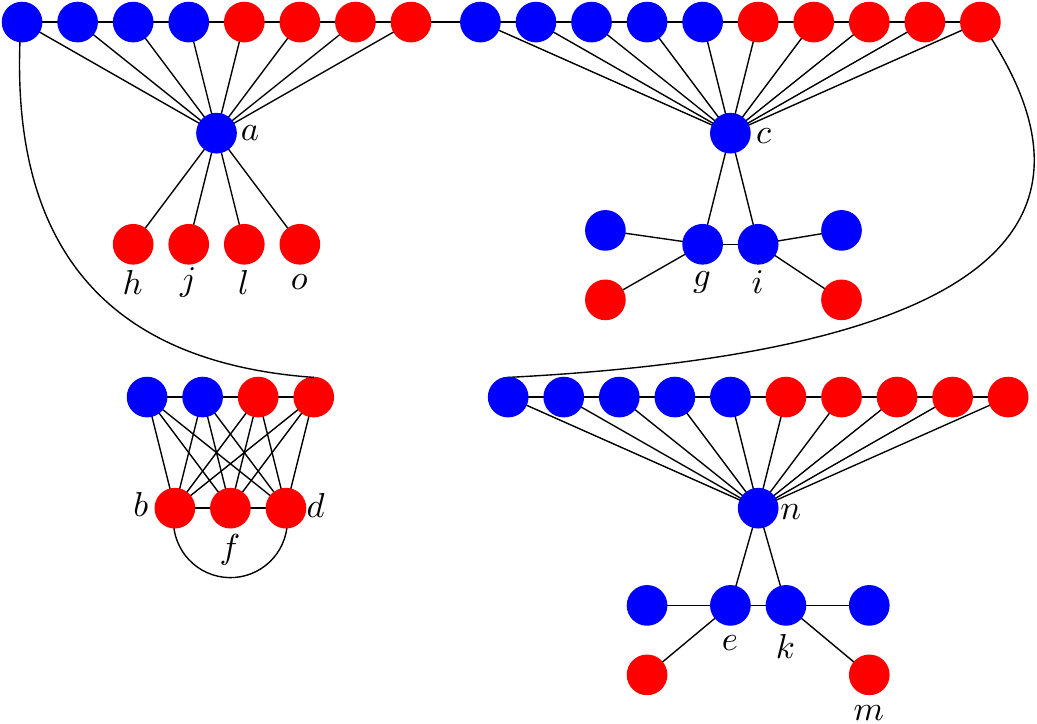}
					\caption{Placement after the last swap of $n$ and $o$.}\label{lastswap}
				\end{subfigure}	
				\begin{subfigure}{0.4\textwidth}
					\centering
					\includegraphics[scale=0.7]{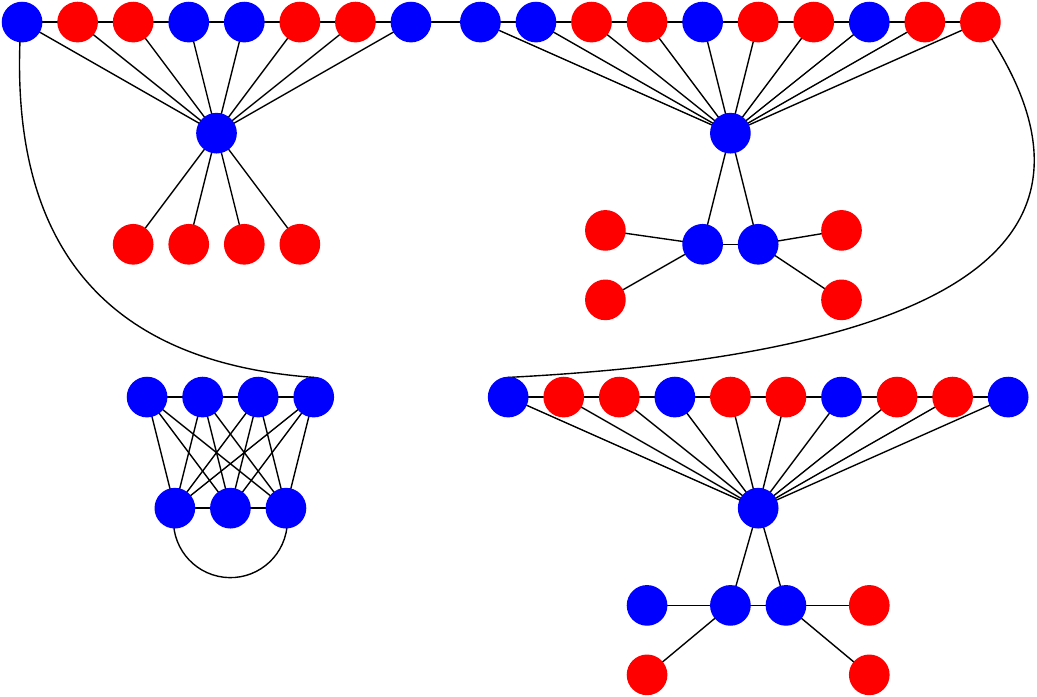}
					\caption{The SE for the same instance.}\label{NE}
				\end{subfigure}	
				\caption{An IRC and the SE for a game $(G,b, \frac12)$.}\label{IRC}
			\end{figure}
		\end{proof}
	
		\noindent For the special case of $\Lambda=1/2$, however, existence of a SE is guaranteed in any graph whose independence number is at least the number of blue agents.
		
		\begin{theorem}\label{thm:bipartite}
			Fix a game $(G,b,\Lambda)$ with $\frac{1}{\delta(G)+1}\leq\Lambda\leq 1/2$. Any strategy profile in which all agents of a same color are located on an independent set of $G$ is a SE.
		\end{theorem}
	
		\begin{proof}
			Let $\sp$ be a strategy profile in which all agents of a same color are located on an independent set of $G$. Assume, w.l.o.g., that all blue agents are assigned to the nodes of an independent set of $G$ and consider a profitable swap performed by a blue agent $i$ and a red agent $j$.
			If $\i=0$, since~$i$ is only adjacent to red agents other than~$j$, it holds that $f_j(\ssp) = 1$, which gives $\u_j(\ssp) = 0$, thus contradicting the fact that $j$ performs a profitable swap.
			If $\i=1$, instead, we obtain $f_i(\sp) = \frac{1}{\delta(\sp(i))+1}\leq\frac{1}{\delta(G)+1} \leq \Lambda$. The numerator comes from the fact that $i$ is only adjacent to red agents. Knowing that $i$ cannot be at the peak, we conclude that she is below the peak. If $j$ is also below the peak, Lemma~\ref{lemma0} contradicts the fact that the swap is profitable, while, if $j$ is above the peak, the contradiction comes from Lemma \ref{lemma0.1}.
		\end{proof}
		
		\begin{corollary}\label{cor:bipartite}
			For $\Lambda=1/2$, games played on bipartite graphs always admit a SE which can be efficiently computed.
		\end{corollary}
	
	\section{Price of Anarchy}\label{sec:PoA}
	
		In this section, we give bounds on the PoA for games played on different topologies, even in those cases for which existence of a SE in not guaranteed.
	
		\paragraph{General Graphs}
			Next lemmas provide a necessary condition that needs to be satisfied by any SE and an upper bound of the value on the social optimum, respectively.
			
			\begin{lemma}\label{lemmar1}
				In a SE for any game $(G,b,\Lambda)$, no agents of different colors can be segregated.
			\end{lemma}
			\begin{proof}
				Fix a strategy profile $\sp$. If there exist two agents $i$ and~$j$ such that $f_i(\sp)=f_j(\sp)=1$, they can perform a profitable swap, as $f_i(\sp)=f_j(\sp)=1$ and $f_i(\ssp)=f_j(\ssp)\notin\{0,1\}$. So, $\sp$ cannot be a SE for $(G,b,\Lambda)$.
			\end{proof}
			
			\begin{lemma}\label{optreg}
				For any game $(G,b,\Lambda)$, we have $$\hn(\sp^*)\leq\min\{(\Delta(G)+1)b,n\}.$$
			\end{lemma}
			\begin{proof}
				As a blue node can be adjacent to at most $\Delta(G)$ red ones, it follows that, in any strategy profile, there cannot be more than $(\Delta(G)+1)b$ non-segregated agents, so that $\hn(\sp^*)\leq\min\{(\Delta(G)+1)b,n\}$.
			\end{proof}
	
			\noindent We now give (almost) tight bounds on the PoA for general graphs.
		
			\begin{theorem}\label{poa-gen}
				For any game $(G,b,\Lambda)$, $$\hpoa(G,b,\Lambda) \leq\min\bigl\{\Delta(G),\frac{n}{b+1}, \frac{(\Delta+1)b}{b+1}\bigr\}.$$ Moreover, there exists a game on a bipartite graph such that $\hpoa(G,b,\Lambda)\geq\frac{n}{b+1}$ when $b > 1$ and $\hpoa(G,b,\Lambda)\geq\frac{n-1}{3}$ when $b=1$.
			\end{theorem}
			\begin{proof}
				For the upper bound, fix a game $(G,b,\Lambda)$ and a SE~$\sp$. By Lemma~\ref{lemmar1}, only agents of one color, say $c$, can be segregated in $\sp$. Thus, we get $\hn(\sp)\geq b+1.$ Let $V$ be the set of nodes of color $c'\neq c$. Every node in $V$ has to be adjacent to a node of color $c$. So, there are at least $|V|\geq b$ non-monochromatic edges in the coloring induced by $\sp$. As every node of color $c$ can be adjacent to at most $\Delta(G)$ nodes of color $c'$, there must be at least $\left\lceil b/\Delta(G)\right\rceil$ nodes of color~$c$ incident to a non-monochromatic edge, that is, being non-segregated in $\sp$. Thus, we get $$\hn(\sp)\geq\frac{(\Delta(G)+1)b}{\Delta(G)}.$$ We conclude that $$\hn(\sp)\geq\max\bigl\{\frac{(\Delta(G)+1)b}{\Delta(G)},b+1\bigr\}.$$ The upper bounds follow from Lemma \ref{optreg}.
				For the lower bounds, consider the games defined in Figure \ref{fig2}.
			\end{proof}
	
			\begin{figure}[h]
				\center
				\begin{subfigure}{0.4\textwidth}
					\centering
					\includegraphics[scale=0.8]{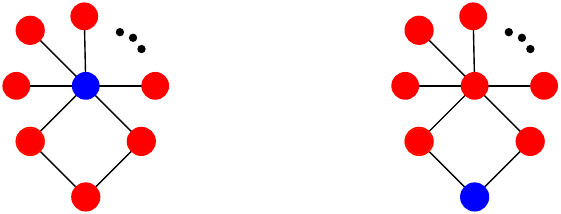}
					\caption{An instance with $b=1$ blue agents. Left: $\sp^*$ with $\hn(\sp^*) = n-1$. Right: a SE $\sp$ with $\hn(\sp) = 3$.}
				\end{subfigure}	\hspace*{2em} 
				\begin{subfigure}{0.4\textwidth} \vspace*{1em}
					\centering
					\includegraphics[scale=0.8]{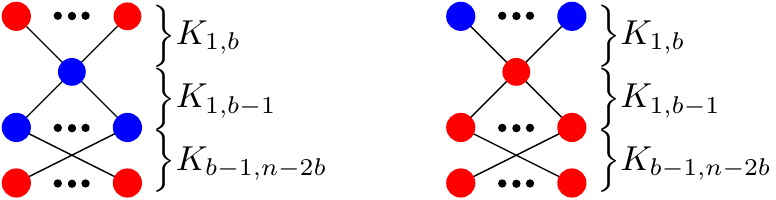}
					\caption{An instance with $b\geq 2$ blue agents. Left: $\sp^*$ with $\hn(\sp^*) = n$. Right: a SE $\sp$ with $\hn(\sp) = b+1$.}
				\end{subfigure}	
				\caption{Lower bounds for $\hpoa(G,b,\Lambda)$ when (a) $b = 1$, and (b) $b > 1$. Left: the socially optimal placement $\sp^*$. Right: the SE $\sp$ with minimum social welfare.}\label{fig2}
			\end{figure}
	
		\paragraph{Regular Graphs}
			For $\delta$-regular graphs, we derive an upper bound of $\delta$ on the PoA from Theorem \ref{poa-gen}. A better result is possible when $\Lambda$ is sufficiently small.
	
			\begin{theorem}\label{ub-reg}
				For any game $(G,b,\Lambda)$ on a $\delta$-regular graph $G$ with $\Lambda<1/\delta$, $\hpoa(G,b,\Lambda)\leq \min\bigl\{(\delta+1)/2,n/2b\bigr\}$.
			\end{theorem}
			
			\begin{proof}
				Fix a SE $\sp$. By Lemma \ref{lemmar1}, only agents of a unique color, say $c$, can be segregated in $\sp$. Let $V$ be the set of nodes of color $c'\neq c$.
				As $\Lambda<1/\delta$, every node in $V$ has to be adjacent to nodes of color $c$ only. Otherwise, any agent in~$V$ that is adjacent to an agent of color $c'$ can perform a profitable swap with a segregated agent of color $c$. Thus, there are at least $\delta|V|\geq\delta b$ non-monochromatic edges in the coloring induced by $\sp$. As every agent of one color can be adjacent to at most~$\delta$ agents of the other one, there are at least~$b$ non-segregated agents of color $c$. Together with the at least~$b$ agents of color~$c'$, this gives $\hn(\sp)\geq 2b$ which, together with Lemma~\ref{optreg}, yields the claim.
			\end{proof}
	
			\noindent As a lower bound, we have the following.
			
			\begin{restatable}{theorem}{theoremlbpoareg}\label{lb-reg}
				For every $\delta\geq 2$ and $\Lambda\leq 1/2$, there exists a game $(G,b,\Lambda)$ on a $\delta$-regular graph such that $$\hpoa(G,b,\Lambda) \geq \frac{\delta(\delta+1)}{2\delta+1}=\frac{\delta+1}{2}-\frac{\delta+1}{4\delta+2}.$$ 
			\end{restatable}
		
			\begin{proof}
			Consider graph $G$ shown in Figure \ref{fig:poaregular}. $G$ is made of three combined gadgets that we call the {\em left gadget}, the {\em upper right gadget} and the {\em lower right gadget}. The left gadget is essentially the complete bipartite graph $K_{\delta,\delta}$ with a missing edge: the one connecting the last two nodes of the respective partitions. These nodes are connected to the upper right gadget and the lower right one, respectively. So each node in the left gadget has degree $\delta$. The upper right gadget consists of a clique $K_{\delta-1}$ whose nodes are all connected to two special nodes, one on the left of $K_{\delta-1}$ and one on the right. The node on the left is the one adjacent to the node from the left gadget, while the node on the right connects the gadget with the lower right one. Thus, every node in this gadget has degree $\delta$. Finally, the lower right gadget is any $\delta$-regular graph with $n'$ nodes with a missing edge: the one connecting the node incident to the edge coming from the left gadget with the node incident to the edge coming from the upper right gadget. So, every node in this gadget has degree equal to $\delta$ and we conclude that $G$ is $\delta$-regular.
			
			Now set $b=\delta$. We claim that the strategy profile $\sp$ depicted in Figure \ref{fig:poaregular} is a SE. If a red agent $i$ swaps with a non-adjacent blue agent $j$, we have $f_i(\ssp)=1$ which results in a non-profitable swap. So a red agent can profitably swap only with an adjacent blue agent. Any blue agent $j$ has $f_j(\sp)=\frac{1}{\delta+1}$. By swapping with an adjacent red agent $i$, we have that either $f_j(\ssp)=\frac{1}{\delta+1}$ or $f_j(\ssp)=\frac{\delta}{\delta+1}$ which never yields an improvement when $\Lambda\leq 1/2$ or $f_j(\ssp)=\frac{\delta-1}{\delta+1}$ when swapping one blue agent $j$ with the rightmost red
			one of the left gadget. In this case, it is red agent $i$ that has no
			improvement in swapping since $f_i(\sp)=\frac{2}{\delta+1}$ and $f_i(\ssp)=\frac{\delta}{\delta+1}$ which never yields an improvement when $\Lambda\leq 1/2$. So, $\sp$ is a SE such that $\hn(\sp)=2\delta+1$.
			
			By letting $n'$ go to infinity, it is always possible to select $b=\delta$ nodes in the lower right gadget such that their closed neighborhoods are pairwise disjoint, which yields $\hn(\sp^*)=\delta(\delta+1)$ and thus the desired lower bound.
			\end{proof}
			\begin{figure}[t]
				\centering
				\includegraphics[width=0.8\linewidth]{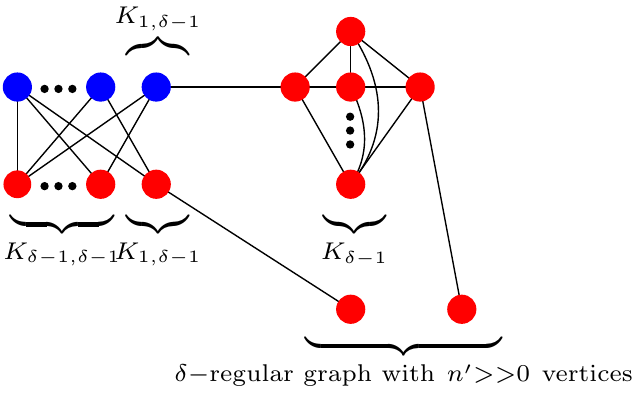}
				\caption{Lower bound construction for the PoA on $\delta$-regular graphs. The SE $\sp$ has $\hn(\sp) = 2\delta+1$.}
				\label{fig:poaregular}
			\end{figure}	
	
			\noindent The lower bound given in Theorem \ref{lb-reg} holds for all values of~$\delta$. It may be the case then that, for fixed values of $\delta$, better bounds are possible. For $\delta=2$ indeed, lower bounds matching the upper bounds given in Theorems \ref{poa-gen} and \ref{ub-reg} can be derived.
			
			\begin{restatable}{theorem}{theorempoarings}
				\label{thm:theorempoarings}
				For any $\epsilon>0$, there exists a game $(G,b,\Lambda)$ on a ring such that $\hpoa(G,b,1/2)>2 -\epsilon$ and $\hpoa(G,b,\Lambda)>3/2 -\epsilon$ for $\Lambda<1/2$.
			\end{restatable}
		
			\begin{proof}
			For any even $b\geq 2$, let $G$ be a ring defined by the sequence of nodes $v_1,v_2,\ldots,v_n$, with $n=3b$.
			
			Assume $\Lambda=1/2$. In this case, $$p(1/3)=p(2/3)>p(0)=p(1).$$ Let $\sp$ be the strategy profile obtained as follows: starting from $v_1$, assign two blue agents followed by a red as long as this is possible. At this time, nodes up to $v_x$, with $x=3b/2$, have been assigned an agent. All the remaining nodes are assigned the remaining red agents. Since for any blue agent $i$ we have $f_i(\sp)=2/3$, all blue agents are getting the largest possible utility and are not interested in swapping. Thus, $\sp$ is a SE. As all agents residing at nodes $v_n,v_1,\ldots,v_x$ are not segregated, we have $\hn(\sp)=3b/2+1$. 
			
			Now assume $\Lambda<1/2$. In this case, $$p(1/3)>p(2/3)>p(0)=p(1).$$ Let $\sp$ be the strategy profile obtained by alternating blue and red agents for as much as possible. Since for any blue agent~$i$ we have $f_i(\sp)=1/3$, all blue agents are getting the largest possible utility and are not interested in swapping. Thus, $\sp$ is a SE. As all agents residing at nodes $v_n,v_1,\ldots,v_{2b}$ are not segregated, we have $\hn(\sp)=2b+1$.
			
			A strategy profile of social value $3b$ can be obtained by sequencing triplets made of two red agents with a blue one in between. Both claims follow by choosing $b$ sufficiently large.
			\end{proof}
	
	\section{Price of Stability}\label{sec:PoS}
	
		In this section, we give bounds on the PoS for games played on different topologies.
	
		\paragraph{General Graphs}
			We give a lower bound on the PoS on general graphs which asymptotically matches the upper bound on the PoA when $b=\Theta(\sqrt{n})$ and $\Lambda$ is a constant w.r.t~$n$.
	
			\begin{restatable}{theorem}{theoremnine}\label{posLB_general}
				For every $\Lambda$, there is a game $(G,b,\Lambda)$ such that $\hpos(G,b) = \Omega(\sqrt{n\Lambda})$.
			\end{restatable}
		
			\begin{proof}
				Let $q\geq 2$ be an integer such that $\frac{1}{q}\leq\Lambda<\frac{1}{q-1}$. Consider the instance in Figure~\ref{fig:pos_general} in which there is a clique of $b$ nodes such that every node in the clique is additionally adjacent to $(q-1)b$ leaves (depicted to the up) and to the leaf of a star with $q$ nodes (depicted at the bottom). 
				\begin{figure}[t]
					\centering
					\includegraphics[width=0.6\linewidth]{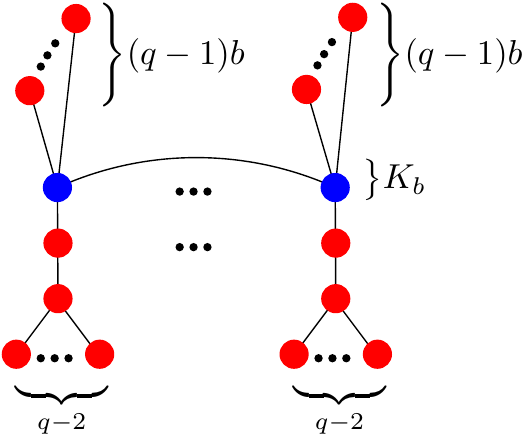}
					\caption{The instance used in the proof of Theorem \ref{posLB_general}. Shown is the socially optimal placement $\sp^*$.}
					\label{fig:pos_general}
				\end{figure}
				It is easy to check that $n=(q-1)b^2+b+bq$. Letting $b$ go to infinity, we get $n = \Theta((q-1)b^2)$, by which $b=\Omega(\sqrt{\frac{n}{q-1}})=\Omega(\sqrt{n\Lambda})$.
				
				The socially optimal placement $\sp^*$, depicted in Figure~\ref{fig:pos_general} has all blue agents on the nodes of the clique and thereby achieves $$\hn(\sp^*) = (q-1)b^2+2b=\Omega((q-1)b^2).$$ 
				
				We claim that, in contrast, any SE $\sp$ can have at most one blue agent in the clique. Assume, by way of contradiction, that there are two or more blue agents in the clique, and let $i$ be one of these agents. Then, there is at least one center of a star of $q$ nodes that is occupied by a segregated red agent~$j$. If $i$ and $j$ swap, we have $f_j(\sp)=1$ and $f_j(\ssp)\notin\{0,1\}$ so that $j$ improves, and $$f_i(\sp)\leq\frac{b}{qb+1}$$ and $$f_i(\ssp)=\frac{1}{q}\in\left(\frac{b}{qb+1},\Lambda\right],$$ so that $i$ improves too. This contradict that $\sp$ is a~SE. 
				Hence, for any SE $\sp$, it holds that $$\hn(\sp)\leq (b-1)q+qb+1=O(2qb)$$ (the fact that a SE exists can be easily checked by considering the strategy profile obtained by placing one blue agent in a node $v$ of the clique and all the remaining $b-1$ ones to the $b-1$ center of a star of $q$ nodes that are not appended to $v$).
				It follows that the price of stability is $$\frac{\Omega((q-1)b^2)}{O(2qb)}=\Omega(b)=\Omega(\sqrt{n\Lambda}).$$
			\end{proof}
	
			\paragraph{Bipartite Graphs}
				For bipartite graphs, we provide a tight bound of $2$ for the PoS of games for which the peak is at $1/2$. We start with the upper bound.
		
				\begin{restatable}{theorem}{theoremten}\label{thm:pos_bipartite}
					For any game $(G,b,1/2)$ on a bipartite graph~$G$, we have $\hpos(G,b,1/2) \leq 2$.
				\end{restatable}
			
				\begin{proof}
					Let $(V_1,V_2)$, with $|V_1|\leq |V_2|$, be the bipartition of the nodes of $G$. For a fixed optimal profile $\sp^*$, denote by~$B_1$ (resp. $B_2$) the set of nodes of $V_1$ (resp. $V_2$) occupied by a blue agent in $\sp^*$. Moreover, denote by $R_1$ (resp. $R_2$) the set of nodes occupied by a red agent in $\sp^*$ falling in the neighborhood of some node in $B_2$ (resp. $B_1$). Clearly, we have $\hn(\sp^*)\leq b+|R_1|+|R_2|$. We shall prove the existence of two swap equilibria, namely $\sp_1$ and $\sp_2$, whose performance compare nicely with that of $\sp^*$.
					
					To construct $\sp_1$, start from $\sp^*$ and swap all blue agents in~$B_2$ with red agents in $V_1$ as long as this is possible. If all blue agents end up in $V_1$, we have that all blue agents occupy the nodes of an independent set of $G$ and so, by Theorem~\ref{thm:bipartite}, $\sp_1$ is an SE. If some blue agents are left out from $V_1$, then all red agents are located in $V_2$. So, we have that all red agents occupy the nodes of an independent set of $G$ and, by Theorem~\ref{thm:bipartite}, $\sp_1$ is an SE also in this case. As the set of nodes in $B_1$ are blue in both $\sp^*$ and $\sp_1$, we obtain that $\hn(\sp_1)\geq |R_2|$.
					
					Equilibrium $\sp_2$ is obtained symmetrically by swapping all blue agents in $B_1$ with red agents in $V_2$ as long as this is possible. In this case, as $b\leq n/2\leq |V_2|$, all blue agents end up in $V_2$ and, by Theorem \ref{thm:bipartite}, $\sp_2$ is an SE. As the set of nodes in $B_2$ are blue in both $\sp^*$ and $\sp_2$ and all blue agents are adjacent to some red agent in $\sp_2$, we obtain that $$\hn(\sp_2)\geq b+|R_1|.$$	
					Thus, we conclude that 
					\begin{eqnarray*}
						\hpos(G,b,1/2)&\leq&\frac{\hn(\sp^*)}{\max\{\hn(\sp_1),\hn(\sp_2)\}} \\
						&\leq&\frac{b+|R_1|+|R_2|}{\max\{b+|R_1|,|R_2|\}}\leq 2.
					\end{eqnarray*}
				\end{proof}
				\noindent We now give the matching lower bound.
		
				\begin{restatable}{theorem}{theoremeleven}\label{thm:pos_bipartite_lb}
					There exists a game $(G,b,1/2)$ on a bipartite graph such that $\hpos(G,b,1/2) \geq 2$.
				\end{restatable}
			
				\begin{proof}
					Consider the instance $(G,b,1/2)$ defined in Figure~\ref{fig4}. $G$ consists of a path of $b$ nodes, that we call the base of the graph. Any node in the base of the graph is additionally connected to $2(b-1)$ leaves (depicted on the up) and to a $2$-node path (depicted on the bottom).
					\begin{figure}[t]
						\center
						\includegraphics[width=0.8\linewidth]{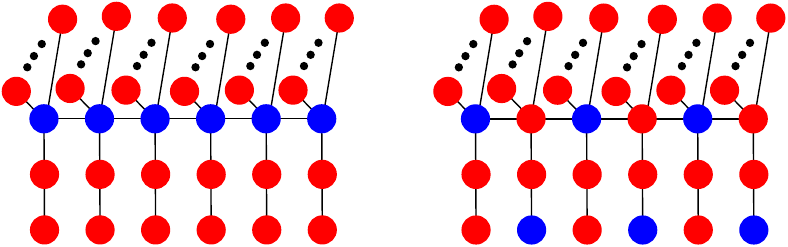}
						\caption{Left: a game with its socially optimal strategy profile $\sp^*$ shown. Right: the SE with maximum social welfare for the same instance.
						}\label{fig4}
					\end{figure}
					For the socially optimal profile~$\sp^*$, we get $$\hn(\sp^*) = 2(b-1)b+2b.$$ 
					However, this is not a SE. In any strategy profile $\sp$ in which two blue agents are adjacent in the base of the graph, like in the socially optimal profile $\sp^*$, one of them, denote this agent by $i$, can swap with a segregated red agent, denoted by~$j$, placed on a leaf node in the lower row. Observe that agent~$j$ is always guaranteed to exists. Since we have $$f_i(\sp)\leq\frac{b}{2b+1}<1/2$$ and $$f_i(\ssp)=1/2,$$ agent $i$ improves her utility. For agent $j$, we have $f_j(\sp)=1$ and $f_j(\ssp)\notin\{0,1\}$, so the swap is profitable and $\sp$ cannot be a SE. The maximum number of agents with non-zero utility that can be obtained by respecting this necessary constraint is $$\hn(\sp) = b(b-1)+\frac{5b}{2},$$ achieved by the SE $\sp$ depicted in Figure~\ref{fig4}~(right). The claim follows by letting $b$ go to infinity. 
				\end{proof}
		
			\paragraph{Almost Regular Graphs}
				We provide upper bounds to the PoS for games played on almost regular graphs. We start by considering the case of graphs with small degree.
				
				\begin{restatable}{theorem}{theoremtwelve}\label{thm:PoS_almost_cubic_graphs}
					For any game $(G,b,\Lambda)$ on an almost regular graph with $\Delta(G)\leq 3$ and $\Lambda\leq 1/2$, $\hpos(G,b,\Lambda)=1$.
				\end{restatable}
			
				\begin{proof}
					Let $\sp^*$ be a socially optimal profile. Using Lemma~\ref{lm:algorithm_for_SE_almost_cubic_graphs} we have that there exists a SE $\sp$ satisfying $\hn(\sp)\geq\hn(\sp^*)$.
				\end{proof}
		
				\noindent An analogous result holds for the case in which $b \geq \alpha(G)$.
				\begin{theorem}\label{thm:pos_for_regular_graphs_and_large_b}
					For any game $(G,b,\Lambda)$ on an almost regular graph with $b \geq \alpha(G)$ and $\frac{1}{\delta(G)+1}\leq\Lambda \leq 1/2$, we have $\hpos(G,b,\Lambda)=1$.
				\end{theorem}
				\begin{proof}
					We prove the claim by showing the existence of a SE~$\sp$ such that $\hn(\sp)=n$. Clearly, if $b=\alpha(G)$, then the strategy profile~$\sp$ in which the blue agents occupy all the nodes of a maximum independent set of $G$ is a SE by Theorem \ref{thm:bipartite}, and thus, $\hn(\sp)=n$, and the statement follows. Therefore, in the following we assume that $b,r > \alpha(G)$.
					
					Let $\sp$ be a strategy profile minimizing the value $\Phi(\sp)$ (ties are arbitrarily broken). By Theorem \ref{thm:conv_regular}, $\sp$ is a SE. 
					We now prove that $\hn(\sp)=n$. For the sake of contradiction, assume that $\hn(\sp)<n$, i.e., there is at least a segregated agent, say~$i$, in $\sp$. Assume w.l.o.g. that $i$ is blue. We claim that all red agents are placed on nodes that form an independent set, i.e., $r \leq \alpha(G)$. This allows us to obtain the desired contradiction as $r>\alpha(G)$. If the red agents are not placed on nodes that form an independent set, then there exists a red agent, say~$j$, having at least a red neighbor in $\sp$. The strategy profile $\ssp$ satisfies $$\Phi(\sp)-\Phi(\ssp)\geq \delta(\sp(i))+1-\left(\delta(\sp(j))-1\right)\geq 1,$$ since $|\delta(\sp(i))-\delta(\sp(j))|\leq 1$. Therefore, $\Phi(\ssp)<\Phi(\sp)$, thus contradicting the fact that $\sp$ minimizes $\Phi$. 
				\end{proof}
		
				\noindent A game $(G,b,\Lambda)$ is {\em balanced} if $b = \left\lfloor n/2\right\rfloor$. Using Theorem~\ref{thm:pos_for_regular_graphs_and_large_b}, we show that the PoS is $1$ in balanced games on regular graphs.
		
				\begin{restatable}{corollary}{corollaryfive}
					For any balanced game $(G,b,\Lambda)$ on a $\delta$-regular graph $G$ and $\frac{1}{\delta+1}\leq \Lambda \leq 1/2$, we have $\hpos(G,b,\Lambda)=1$.
				\end{restatable}
				\begin{proof}
					We have that $b = \lfloor n/2 \rfloor$. We show that $\alpha(G) \leq \lfloor n/2 \rfloor$ using a simple counting argument. This allows us to use Theorem \ref{thm:pos_for_regular_graphs_and_large_b} to prove the claim. 
					
					To show the upper bound on $\alpha(G)$, we count all the edges that are incident to the nodes of a fixed maximum independent set of $G$ and bound this value from above by the number of edges of the graph, thus obtaining the following inequality $\delta\alpha(G) \leq \frac{\delta}{2} n$, i.e., $\alpha(G) \leq n/2$. 
					Using the fact that $\alpha(G)$ is an integer value, we derive $\alpha(G) \leq \lfloor n/2 \rfloor$.
				\end{proof}
		
				\noindent Before proving that the PoS is $O(1/\Lambda)$ when $b<\alpha(G)$, we need to introduce some new definitions and additional technical lemmas based on some well-known optimization cut problems.	
		
				For a given graph $G$ and a subset of nodes of $V$, we denote by $G[U]$ the subgraph of $G$ {\em induced} by $U$. More precisely, the node set of $G[U]$ is $U$ and, for every $u,v \in U$, $G[U]$ contains the edge $(u,v)$ iff $G$ contains the edge $(u,v)$.
		
				The \textsc{$k$-Max-Cut} problem  is an optimization problem in which, given a graph $G$ as input, we want to compute a $k$-partition $\{V_1,\dots,V_k\}$ of the nodes of $G$ that maximizes the number of edges that cross the cut induced by the $k$-partition, that is, the number of edges $(u,v)$ such that $u \in V_t$, $v \in V_h$, and $h\neq t$.	It is well-known that the greedy algorithm for the  \textsc{$k$-Max-Cut} problem computes a $k$-partition $\{V_1,\dots,V_k\}$ of the nodes of a graph such that, for every node $v$, the number of edges incident to $v$ that cross the cut induced by the $k$-partition is at least $\lceil (1-\frac{1}{k})\delta(v)\rceil$ \citep{vazirani}. Using this folklore result, we can derive the following useful lemma.
		
				\begin{observation}\label{obs:max_cut_greedy_algorithm}
					Let $G$ be a graph and $U \subseteq V$ such that $|U|\geq k$. There exists a $k$-partition $\{V_1,\dots,V_k\}$ of $G[U]$ such that, for every $t\in\{1,\dots,k\}$, the degree of each vertex $v \in V_t$ in $G[V_t]$ is at most $\lfloor \frac{\delta(v)}{k} \rfloor$.
				\end{observation}
				\begin{proof}
					The greedy algorithm for the \textsc{$k$-Max-Cut} problem computes a $k$-partition $\{V_1,\dots,V_k\}$ of $G[U]$ such that, for every node $v$ of $U$, the number of edges incident to~$v$ that cross the cut induced by the $k$-partition is at least $$\lceil (1-\frac{1}{k})\delta(v)\rceil.$$ As a consequence, for any $v \in V_t$, with $t\in \{1,\dots,k\}$, the number of edges that are incident to~$v$ in $G[V_t]$ is at most $$\delta(v)-\lceil (1-\frac{1}{k})\delta(v)\rceil=\lfloor \delta(v)/k \rfloor.$$
				\end{proof}
		
				\noindent The \textsc{Balanced $k$-Max-Cut} problem is a \textsc{$k$-Max-Cut} in which we additionally require the $k$-partition $\{V_1,\dots,V_k\}$ to be {\em balanced}, i.e., for every $t \in \{1,\dots,k\}$, $|V_t|\geq \lfloor n/k \rfloor$.\footnote{When $n < k$, we might have empty sets in the $k$-partition.} For the \textsc{Balanced $k$-Max-Cut} problem we can prove a useful lemma that is analogous to Observation \ref{obs:max_cut_greedy_algorithm}.
		
				\begin{lemma}\label{lm:balanced_max_cut_greedy_algorithm}
					Let $G$ be a graph and $U \subseteq V$  such that $|U|\geq 2$. There exists a balanced $k$-partition $\{V_1,\dots,V_k\}$ of $U$ such that, for at least one $t \in\{1,\dots,k\}$, $|V_t|\geq 1$ and the degree of every $v \in V_t$ in $G[V_t]$ is at most $\lfloor \delta(v)/k \rfloor$.
				\end{lemma}
				\begin{proof}
					In the remainder of this proof, for a given $k$-partition $\{V_1,\dots,V_k\}$ of $U$ and a bijective function $$\rho : \{v_1,\dots,v_k\} \rightarrow \{v_1,\dots,v_k\},$$ with $v_t \in V_t$ for every $t \in \{1,\dots,k\}$, we define the $\rho$-swap as the $k$-partition $\{V_1',\dots,V_k'\}$ in which, for every $t\in\{1,\dots,k\}$, the set $V_t'$ is obtained from $V_t$ by replacing $v_t$ with $\rho(v_t)$ (it may happen that $\rho(v_t)=v_t$). We say that the $\rho$-swap is profitable if the number of edges crossing the cut induced by $\{V_1',\dots,V_k'\}$ is strictly larger than the number of edges crossing the cut induced by $\{V_1,\dots,V_k\}$. A balanced $k$-partition $\{V_1,\dots,V_k\}$ is stable if there is no profitable $\rho$-swap.
					
					Let $\{V_1,\dots,V_k\}$ be a balanced $k$-partition of $U$ that maximizes the number of edges in the cut (ties can be arbitrarily broken). We claim that $\{V_1,\dots,V_k\}$ satisfies all the properties of the lemma statement. This is clearly true when $|U|\leq k$ as all edges of $G$ are in the cut (consider, for instance, the solution in which each set of the $k$-partition contains at most one node). Therefore, we only need to prove the claim when $|U|> k$. This implies that $|V_t|>0$ for every $t\in\{1,\dots,k\}$. As a consequence, we only need to prove that there exists $t \in \{1,\dots,k\}$ such that, for every node $v \in V_t$, the degree of $v$ in $G[V_t]$ is at most $\lfloor \delta(v)/k \rfloor$. 
					
					For the sake of contradiction, assume that there are $k$ nodes $v_1,\dots,v_k$, with $v_t \in V_t$, such that, for every $t\in\{1,\dots,k\}$, the degree of $v_t$ in $G[V_t]$ is strictly larger than $\lfloor \delta(v_t)/k \rfloor$. Consider the graph $H$ on the $k$ nodes $v_1,\dots,v_k$, where we add the direct edge $(v_t,v_h)$ between $v_t$ and $v_h$, with $t\neq h$, iff the number of edges incident to $v_t$ whose other endpoints are in $V_h$ is at most $\lfloor \delta(v_t)/k \rfloor$. We observe that the outdegree of each node in $H$ is at least 1. As a consequence $H$ contains a directed cycle $C$. We define $$\rho:\{v_1,\dots,v_k\} \rightarrow \{v_1,\dots,v_k\}$$ as follows. For each edge $(v_t,v_h)$ of $C$, we define $$\rho(v_t)=v_h,$$ while $$\rho(v_\ell)=v_\ell$$ for every $\ell \in \{1,\dots,k\}$ such that $v_\ell$ is not contained in $C$. Clearly, $\rho$ is a bijective function. Furthermore, the the $\rho$-swap is profitable, as each node $v_t$ that is moved from $V_t$ to $V_h'$, with $h \neq t$, contributes with at least one more edge in the cut. This contradicts the fact  that $\{V_1,\dots,V_k\}$ maximizes the number of edges in the cut.
				\end{proof}
				
				\noindent We are now ready to prove the upper bound on the PoS for games played on almost regular graphs when $b < \alpha(G)$.
		
				\begin{restatable}{theorem}{theoremfourteen}\label{thm:pos_for_regular_graphs_and_small_b}
					For any game $(G,b,\Lambda)$ on an almost regular graph $G$ with $b < \alpha(G)$ and $\Lambda \leq 1/2$, we have $\hpos(G,b, \Lambda)= \min\{\Delta(G)+1,\text{O}(1/\Lambda)\}$.
				\end{restatable}
		
				\begin{proof}
				Let $\Delta=\Delta(G)$. Let $\sp^*$ be a strategy profile that maximizes $\hn(\sp^*)$ for $(G,b, \Lambda)$. Let $B$ and $R$ be the nodes occupied by the non-segregated blue and red agents in $\sp^*$, respectively. 
				Clearly, $\hn(\sp^*)=|B|+|R|$. Moreover, $|B|\geq 1$ iff $|R|\geq 1$. We prove the claim by showing the existence of a strategy profile $\sp$ such that $\sp$ is a SE and $$\hn(\sp) = \Omega(\Lambda(|B|+|R|)).$$
				
				We first rule out the case in which $b \leq \lfloor \Lambda(\Delta-1) \rfloor-1$. In fact, in this case $G[B]$ has maximum degree of at most $\lfloor \Lambda(\Delta-1)\rfloor-1$, which implies that every blue agent is below the peak in $\sp^*$. As a consequence, from Lemma \ref{lemma0} and Lemma~\ref{almost}, $\sp^*$ is also a SE. Therefore, in the following we assume that $b > \lfloor \Lambda(\Delta-1) \rfloor-1$. As $b$ is an integer, we have that $b \geq \Lambda\Delta-\Lambda- 1$, from which we derive
				\begin{equation}\label{eq:b_vs_delta_ratio}
					\Delta \leq b/\Lambda
					+1+1/\Lambda.
				\end{equation}
				
				Next, we rule out the case in which $|B|+|R| = O( b/\Lambda)$, i.e., $b =\Omega(\Lambda(|B|+|R|))$. In fact, in this case let $\sp$ be any SE (that we know to exist by Theorem \ref{thm:conv_regular}). By Lemma \ref{lemmar1}, there is a color for which all agents of that color are not segregated in~$\sp$. Therefore $$\hn(\sp) \geq \min\{b,r\}=b= \omega(\Lambda(|B|+|R|)).$$ Hence, for the rest of the proof, we assume that $|B|+|R| = \omega(b/\Lambda)$. As $|B| \leq b$, it holds 
				\begin{equation}\label{eq:lb_R}
					|R|= \omega(b/\Lambda). 
				\end{equation}
				Finally, we rule out the case in which $\Delta = O(1/\Lambda)$. By Theorem \ref{poa-gen}, $\hpos(G,b,\Lambda) \leq \Delta+1$.
				In the following, for any subset $B'$ of $B$, we denote by ${\mathcal R}(B')$ the subset of nodes of $R$ that are {\em dominated} by $B'$, i.e., ${\mathcal R}(B')$ contains all nodes $u \in R$ for which there exists a node $v$ in $B'$ such that $(u,v)$ is an edge of $G$.
				
				We iteratively use Observation \ref{obs:max_cut_greedy_algorithm} to compute hierarchical $k$-partitions of $B$, with $k=\lceil \frac{\Delta-1}{\Lambda\Delta-1} \rceil$. We observe that $k\geq 1$ as $\Lambda > 1/(\Delta+1)$.
				
				Starting from $B$, we compute the $k$-partition $\{B_0^1,\dots,B_0^k\}$ of $B$ that satisfies the premises of Observation \ref{obs:max_cut_greedy_algorithm}. This is called the $k$-partition of {\em level $0$}. For the remaining part of this proof we use the shortcut $R_0^t$  to denote ${\mathcal R}(B_0^t)$. W.l.o.g., we assume that $|R_0^1|\geq \max_t |R_0^t|$. Now, given the $k$-partition $\{B_{h}^1,\dots,B_{h}^k\}$ of level $h$ such that $|B_{h}^1|\geq 2$ and $|R_h^1|\geq \max_t |R_h^t|$, we compute a $k$-partition $\{B_{h+1}^1,\dots,B_{h+1}^k\}$ of $B_{h}^1$ that satisfies the premises of Observation \ref{obs:max_cut_greedy_algorithm}. This is the $k$-partition of {\em level $h+1$} where, again, we use the shortcut $R_{h+1}^t$  to denote ${\mathcal R}(B_{h+1}^t)$ and, w.l.o.g., we assume that $|R_{h+1}^1|\geq \max_t |R_{h+1}^t|$. Clearly, the $k$-partition $\{B_L^1,\dots,B_L^k\}$ computed in the last level, say $L$, satisfies $|B_L^1|= 1$ and $|R_L^1| \leq \Delta$. We observe that $$|R|\geq |R_0^1|\geq |R_1^1|\geq \dots \geq |R_L^1|$$ by construction.
				
				Let $\ell$ be the minimum index such that $|R\setminus R_\ell^1|\geq k(b-1)$. Such an index always exists because, using inequality, (\ref{eq:lb_R}), inequality (\ref{eq:b_vs_delta_ratio}), and the fact that $\Delta =\omega(1/\Lambda)$, we have
				$$|R\setminus R_L^1|=|R|-|R_L^1|\geq|R|-\Delta= \omega(b/\Lambda) \geq kb.$$
			
				Let $\{R^1,\dots,R^k\}$ be a balanced $k$-partition of $R\setminus R_\ell^1$ that satisfies the premises of Lemma \ref{lm:balanced_max_cut_greedy_algorithm}. We have that $|R^t|\geq b$ for every $t \in \{1,\dots,k\}$. W.l.o.g., we assume that the degree of each node $v \in R^1$ in $G[R^1]$ is at most $\lfloor\delta(v)/k \rfloor$. By construction, also each node $v \in B_\ell^1$ has a degree in $G[B_\ell^1]$ of at most $\lfloor \delta(v)/k\rfloor$. Since there is no edge $(u,v)$ of $G$ such that $u \in B_\ell^1$ and $v \in R\setminus R_\ell^1$, it follows that the degree of each node $v \in B_\ell^1 \cup R^1$ in $G[B_\ell^1 \cup R^1]$ is also upper bounded by $\lfloor \delta(v)/k \rfloor$.
			
				Let $\sp$ be the strategy profile in which exactly $|B_\ell^1|$ blue agents are placed on the nodes of $B_\ell^1$ and the remaining $b-|B_\ell^1|$ blue agents are placed on a subset of nodes of $R^1$ (ties among nodes of $R^1$ are arbitrarily broken). Let $i$ be any blue agent ane let $v$ be the node occupied by $i$ in $\sp$. By the choice of $k$ we have $k \geq \lceil\frac{\delta(v)}{\Lambda(\delta(v)+1)-1}\rceil$. This implies that
				$$
				f_i(\sp)\leq \frac{\lfloor \delta(v)/k \rfloor +1}{\delta(v)+1}\leq \frac{\delta(v)/k+1}{(\delta(v)+1)}\leq \Lambda.
				$$
				Therefore, the blue agent $i$ is below the peak in $\sp$. Since every blue agent is below the peak in $\sp$, from Lemma \ref{lemma0} and Corollary \ref{almost}, $\sp$ is a SE. 
				
				We conclude the proof by showing that $$\hn(\sp) = \Omega(\Lambda(|B|+|R|)).$$ By construction, each blue agent occupies a node that is adjacent to at least one other node occupied by a red agent. Moreover, each red agent that occupies a node of $R_\ell^1$ is in the neighborhood of at least one node of $B_\ell^1$ that is occupied by a blue agent). Therefore, 
				\begin{equation}\label{eq:doi_regular_graphs}
					\hn(\sp) \geq b+|R_\ell^1|.
				\end{equation}
		
				For proof convenience, let us denote by $R_{-1}^1$ the set $R$. By the choice of $\ell$ we know that $|R|-|R_{\ell-1}^1|=|R\setminus R_{\ell-1}^1| < kb$, which implies that $|R_{\ell-1}^1| > |R|-kb$. As a consequence, since $$\sum_{i=1}^k|R_\ell^i| \geq |R_{\ell-1}^1|> |R|-kb$$ and $|R_\ell^1| \geq \max_{i}|R_\ell^i|$, we obtain
				\begin{equation}\label{eq:useful_bound_doi_regular_graphs}
					|R_\ell^1|> \frac{|R|}{k}-b.
				\end{equation}
				Combining the inequalities (\ref{eq:doi_regular_graphs}) and (\ref{eq:useful_bound_doi_regular_graphs}) we obtain $$\hn(\sp) > |R|/k.$$ Using inequality (\ref{eq:lb_R}) and the fact that $k=O(1/\Lambda)$ we finally obtain $\hn(\sp) = \Omega(\Lambda(|B|+|R|))$, as desired.
			\end{proof}
		
			\noindent We can derive the following upper bound to the $\hpos$.
		
			\begin{corollary}\label{cor:pos_almost_regular}
				For any game $(G,b,\Lambda)$ on an almost regular graph with a constant value of $\Lambda \leq 1/2$, we have $\hpos(G,b,\Lambda) =\text{O}(1)$.
			\end{corollary}
			\begin{proof}
				By Theorem \ref{poa-gen}, the $\hpos$ is constant if $\Delta(G)$ is constant. The result when $\Delta(G)$ is not constant is divided into two cases. For the case $b \geq \alpha(G)$ the claim immediately follows from Theorem \ref{thm:pos_for_regular_graphs_and_large_b}. For the case $b < \alpha(G)$ the claim follows from Theorem \ref{thm:pos_for_regular_graphs_and_small_b} and the fact that $\Lambda$ is constant by assumption.
			\end{proof}

	\section{Computational Complexity Results}\label{sec:complexity}
		In this section we analyze the computational complexity aspects of the game played on both bipartite and regular graphs. More precisely, we provide hardness results for the two problems of computing a social optimum and a SE $\sp$ that maximizes the value $\hn(\sp)$, respectively. 

		\begin{theorem}\label{thm:reduction_dominating_set}
			There is a constant $c>1$ such that, given a game $(G,b,\Lambda)$ played on a cubic graph $G$, the problem of computing a social optimum strategy profile is not $c$-approximable in polynomial time, unless $\mathsf{P}=\mathsf{NP}$.
		\end{theorem}
		\begin{proof}
			The reduction is from \textsc{Minimum Dominating Set} problem on cubic graphs, an optimization problem in which the goal is to compute a minimum-size set $D$ of nodes of a given cubic graph $G'$ that {\em dominates} $V(G')$, i.e., for every node $v \in V(G')$, $v \in D$ or there is an edge $(u,v) \in E(G')$ such that $u \in D$. It is known that a minimum dominating set on cubic graphs is not approximable within some constant $c' > 1$, unless $\mathsf{P}=\mathsf{NP}$, see \citep{alimonti97}.
			
			Let $G$ be a cubic graph of $n$ nodes that has a minimum dominating set of size $k^*$ and let $b=k^*$. We claim that a strategy profile $\sp^*$ satisfies $\hn(\sp^*)=n$ iff the $b$ blue agents are placed on the nodes that form a minimum dominating set of $G$. Indeed, a red agent placed on a node $v$ is not segregated in $\sp^*$ iff there is a blue agent placed on a node that dominates~$v$. Furthermore, a blue agent placed on a node $u$ is never segregated in $\sp^*$ because of the minimality of the dominating set (i.e., each node of a minimum dominating set $D$ must dominate a node of the graph that is not in $D$).
			
			Let $c=\frac{4}{5-c'}$. Since all the nodes of the graph trivially form a dominating set of size $n\leq 4k^*$ (each node of the dominating set dominates 4 nodes), we have that $c'< 4$ and therefore, $c>1$. 
			
			We complete the proof by showing that, if we were able to compute, in polynomial time, a strategy profile $\sp$ such that $$\frac{\hn(\sp^*)}{\hn(\sp)}\leq c,$$ then we could compute, in polynomial time, a $c'$-approximate dominating set of $G$. 
			
			Let $\sp$ be a strategy profile such that $\frac{\hn(\sp^*)}{\hn(\sp)}\leq c$ and let 
			$D$ be the set of nodes that are occupied by the blue agents in  $\sp$.  Let $n'$ be the nodes of $G$ that are not dominated by $D$. We have that $\frac{n}{n-n'} \leq c$, from which we derive that 
			$$
			n'\leq \frac{c-1}{c}n = \frac{\frac{4}{5-c'}-1}{\frac{4}{5-c'}}n \leq \frac{c'-1}{4}4k^* = (c'-1)k^*.
			$$ 
			We now compute, in polynomial time, a dominating set $D'$ of~$G$ whose size is at most $k^*+(c'-1)k^*=c'k^*$. $D'$ contains~$D$ and all the $n'$ nodes of $G$ that are not dominated by~$D$. Clearly, $D'$ is a dominating set of $G$ that approximates the value $k^*$ within a factor of $c'$. This completes the proof.
		\end{proof}

		\begin{figure}
			\center
			\begin{subfigure}{\linewidth}
				\centering
				\includegraphics[width=1.8cm]{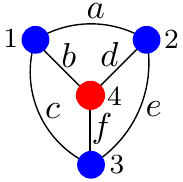}
				\caption{A cubic graph $G'$ with $n' = 4$, $m' = 6$, with a minimum vertex cover $k^* = 3$.  }
			\end{subfigure}	
			\begin{subfigure}{\linewidth}
				\centering
				\includegraphics[width=\linewidth]{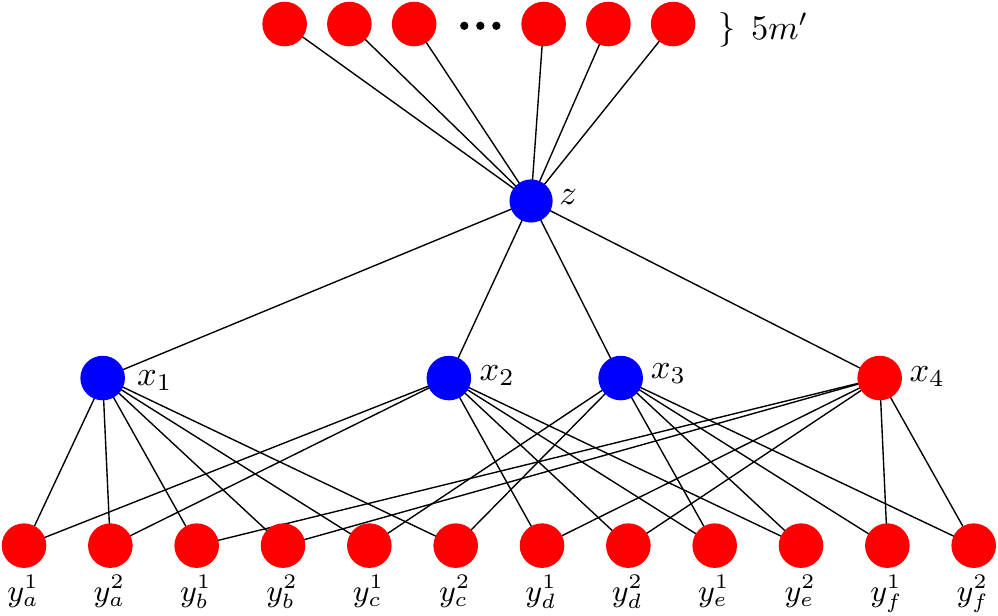}
				\caption{An instance $(G,b)$ of our game constructed from $G'$ with $n = n'+ 7m'+1$ vertices and $b=k^*+1$ blue agents.}
			\end{subfigure}	
			\caption{An example instance of the reduction from Vertex Cover shown in Theorem~\ref{thm:reduction_vertex_cover}.}\label{fig:vertexcover}
		\end{figure}

		\noindent The following lemma allows us to convert any strategy profile into an SE without increasing the number of segregated agents.
		
		\begin{lemma}\label{lm:algorithm_for_SE_almost_cubic_graphs}
			Given a game $(G,b,\Lambda)$ on an almost regular graph, with $\Delta(G)\leq 3$ and  $\Lambda \leq 1/2$, and given a strategy profile $\sp$, we can compute an SE $\sp'$ such that $\hn(\sp') \geq \hn(\sp)$ in polynomial time.
		\end{lemma}
		\begin{proof}
			We shall prove the following claim: if a feasible strategy profile $\sp$ is not a SE, then there exists a (not necessarily profitable) swap decreasing the potential function $\Phi$ and not creating new segregated agents. This implies that after a sequence of at most $m=|E(G)|$ swaps of this type, we obtain an SE $\sp'$ such that $\hn(\sp')\geq \hn(\sp)$. Therefore, given $\sp$, we have that $\sp'$ can be computed in polynomial time.
			
			It remains to prove the existence of a (not necessarily profitable) swap $\sp'$ such that $\Phi(\sp')<\Phi(\sp)$ and not creating new segregated agents. Towards this end, fix a non-equilibrium feasible strategy profile $\sp$ and consider a blue agent $i$ and a red agent $j$ possessing a profitable swap in $\sp$. If no segregated agents are created in $\ssp$, then the claim holds. So assume that a segregated agent $k$ is created in $\ssp$. Clearly, by definition of profitable swaps, it must be $k\notin\{i,j\}$. Assume, w.l.o.g, that $k$ is red. Then, since we have $$f_k(\ssp)=\frac{\delta(\sp(k))+1}{\delta(\sp(k))+1}=1,$$ $k$ needs be adjacent to $i$ in $\sp$, i.e. $\k=1$, and $$f_k(\sp)=\frac{\delta(\sp(k))}{\delta(\sp(k))+1}.$$ Let $x_b$ (resp. $x_r$) be the number of blue agents (resp. red agents other than $k$) adjacent to $i$ in~$\sp$. Since profitable swaps in almost regular graphs can only occur between agents above the peak, we have $$\frac{x_b+1}{x_b+x_r+2}>\Lambda$$ which implies $x_r<x_b$ as $\Lambda \leq 1/2$. By swapping $i$ and~$k$, we get $$\Phi(\sp)-\Phi(\sspp)=\delta(\sp(k))-1+x_b-x_r>0.$$ Therefore~$\sspp$ is a swap such that $\Phi(\sspp)<\Phi(\sp)$.  
			
			We are left to prove that no segregated agents are created in $\sspp$. The neighborhood of node $\sp(k)$ in $\sp$ is composed of node $\sp(i)$ and a remaining set of red nodes. Thus, when $\sp(k)$ and $\sp(i)$ exchange their colors in $\sspp$, no segregated agents are created in the closed neighborhood of $\sp(k)$. The neighborhood of node $\sp(i)$ in $\sp$ is composed of node $\sp(k)$ and a remaining set of $x_b$ blue nodes and $x_r$ red nodes (not counting $\sp(k)$), with $x_b>x_r$. As the maximum degree of $G$ is at most $3$ and $\sp(i)$ is adjacent to $\sp(k)$, we have $x_b+x_r\leq 2$, which, since $x_b>x_r$, implies $x_r=0$. Thus, when $\sp(k)$ and $\sp(i)$ exchange their colors in $\sspp$, no segregated agents are created in the closed neighborhood of $\sp(i)$. No other nodes are affected by the swap, thus no segregated agents are created.
		\end{proof}

		\begin{corollary}
			There is a constant $c>1$ such that, given a game $(G,b,\Lambda)$ on a cubic graph $G$, with $\Lambda \leq 1/2$, the problem of computing an SE $\sp$ that maximizes $\hn(\sp)$ is not $c$-approximable in polynomial time, unless $\mathsf{P}=\mathsf{NP}$.
		\end{corollary}		
		\begin{proof}
			Let $\sp^*$ be a strategy profile that maximizes the value $\hn(\sp^*)$. Thanks to Lemma \ref{lm:algorithm_for_SE_almost_cubic_graphs}, we know that there is an SE $\sp$ such that $\hn(\sp) \geq \hn(\sp^*)$. As a consequence, any SE that approximates $\hn(\sp)$ within a factor of $c$ would also approximate $\hn(\sp^*)$ within a factor of $c$. The claim now follows from Theorem \ref{thm:reduction_dominating_set}. 
		\end{proof}
		
		\noindent We now provide analogous results for bipartite graphs.
		
		\begin{theorem}\label{thm:reduction_vertex_cover}
			There is a constant $c>1$ such that, given a game $(G,b,\Lambda)$ on a bipartite graph $G$, the problem of computing a social optimum strategy profile is not $c$-approximable in polynomial time, unless $\mathsf{P}=\mathsf{NP}$.
		\end{theorem}
		\begin{proof}
			The reduction is from \textsc{Minimum Vertex Cover} problem on cubic graphs, an optimization problem in which the goal is to compute a minimum-size set $C$ of nodes of a given cubic graph $G'$ such that every edge $(u,v)$ of $G'$ is {\em covered} by $C$, i.e., $\{u,v\} \cap C \neq \emptyset$. It is known that a minimum vertex cover on cubic graphs is not approximable within some constant $c' > 1$, unless $\mathsf{P}=\mathsf{NP}$, see \citep{alimonti97}.
			
			Let us assume that $n'$ and $m'=\frac{3}{2}n'$ is the number of nodes and edges of the input graph $G'$, respectively. We construct a graph $G$ as follows (see Figure \ref{fig:vertexcover} for an example). $G$ contains $n=n'+7m'+1$ nodes. More precisely, each node $v$ of $G'$ is modeled by a node $x_v$ in $G$, while each edge $e$ of $G'$ is modeled by two nodes $y_e^1$ and $y_e^2$ in $G$. $G$ also contains a special node $z$ and $5m'$ additional dummy nodes. The special node~$z$ is connected by an edge to each of the $5m'$ dummy nodes and the $n'$ nodes $x_v$, with $v$ being a node of $G'$. Finally,~$G$ contains the two edges connecting $x_v$ with $y_e^1$ and $y_e^2$ iff $v$ is an endpoint of the edge $e$ in $G'$. 
			
			By construction, we have that $G$ is a bipartite graph. Let~$k^*$ denote the size of a minimum vertex cover of $G'$. We consider the game $(G,b)$ played on the constructed graph $G$, where $b=k^*+1$.
			
			We claim that a social-optimal strategy profile $\sp^*$ has a $\hn(\sp^*)=n$ iff $G'$ admits a vertex cover of size $k^*$.
			
			$(\Leftarrow)$ Let $C^*$ be a vertex cover of $G'$ of size $k^*$. Consider the strategy profile $\sp^*$ in which one blue agent is placed on the special node $z$, while the remaining $k^*$ blue agents are placed on the nodes $x_v$, with $v \in C^*$. Clearly, the red agents are placed on the remaining nodes of the graph. By construction, one can check that no agent in $G$ is segregated (see Figure \ref{fig:vertexcover} for an example). Therefore, $\hn(\sp^*)=n$.
			
			$(\Rightarrow)$ Let $\sp^*$ be a strategy profile such that $\hn(\sp^*)=n$. First of all, as $k^*+1 \leq n'+1< 5m'$, we have that no dummy node can be occupied by a blue agent. This is because all edges that connect a dummy node with the special node $z$ must be bi-chromatic and the number of blue agents is not sufficient to cover all the dummy nodes. Therefore, all dummy nodes must be occupied by red agents and, as a consequence, the special node $z$ is occupied by a blue agent. We claim that the set $$ C(\sp^*):=\{v \in V(G') \mid x_v \textrm{ is occupied by a blue agent}\}	$$ has size $k^*$ and forms a vertex cover of $G'$. We observe that, by construction, it is enough to prove that $C(\sp^*)$ has size $k^*$ as each node $y_e^i$, with $i \in \{1,2\}$, is adjacent to the nodes $x_v$ such that $v$ covers $e$ in $G'$. For the sake of contradiction, assume that $|C(\sp^*)|<k$. We show the existence of a vertex cover of $G'$ of size strictly smaller than $k^*$. Let $E'$ be the subset of the edges of $G'$ such that, for each $e \in E'$, $y_e^1$ and $y_e^2$ are both occupied by blue agents. Let $C$ be a set of nodes of~$G'$ that contains $C(\sp^*)$ plus any of the two endnodes of $e$, for each $e \in E'$. We now show that $|C|<k^*$. As $\hn(\sp^*)=n$, each node $y_e^i$ that is occupied by a red agent should be adjacent to a node $x_v$ occupied by a blue agent. By construction~$v$ covers $e$ and $v \in C(\sp^*)$. Therefore, $C(\sp^*)$ covers all the edges of $E(G)\setminus E'$. Hence, $C$ is a vertex cover of $G'$ of size strictly smaller than $k^*$.
			
			We complete the proof by showing that there is a constant $c>1$ such that the problem of computing, in polynomial time, a strategy profile $\sp$ that approximates the social optimum $\sp^*$ is not approximable within $c$, unless $\mathsf{P}=\mathsf{NP}$. Let $c' > 1$ be the constant such that the \textsc{Minimum Vertex Cover} problem on cubic graphs is not approximable within a factor of $c'$ in polynomial time. As each node of~$G'$ covers $3$ edges, we have that $k^* \geq \frac{m'}{3}=\frac{n'}{2}$. This implies that the \textsc{Minimum Vertex Cover} problem on cubic graphs is approximable within a factor of 2 (all nodes of the graph suffice). Therefore, $c'<2$. We set $c=13/(14-c')$. Observe that $c>1$ as $1< c'<2$. We now prove that if we were able to  compute, in polynomial time, a strategy profile $\sp$ such that $\hn(\sp^*)/\hn(\sp) \leq c$, then we could $c'$-approximate the \textsc{Minimum Vertex Cover} problem on cubic graphs in polynomial time.  
			
			For the sake of contradiction, let $\sp$ be strategy profile that $c$-approximates $\hn(\sp^*)$ and assume that $\sp$ can be computed in polynomial time. We use $\sp$ to define a new strategy profile $\sp'$ such that (i) $\hn(\sp')\geq \hn(\sp)$, (ii) one blue agent is placed on the special node $z$, and (iii) all the other blue agents are placed on a subset of nodes $x_v$ with $v \in V(G')$.
			
			First of all, we show that the special node $z$ is occupied by a blue agent in $\sp$. If not, there would be at least $$5m'-k^*-1 \geq 4m'$$ dummy nodes occupied by red agents and therefore $\hn(\sp) \leq n-4m'$. As a consequence, using also the fact that $m'=\frac{3}{2}n'$ and $m'\geq 1$, we would obtain
			\begin{align*}
				\frac{\hn(\sp^*)}{\hn(\sp)} & \geq \frac{n}{n-4m'}=1+\frac{4m'}{n-4m'}=1+\frac{4m'}{n'+3m'+1}\\
				& \geq 1+\frac{3m'}{\frac{2}{3}m'+3m'+m'}=\frac{23}{14}>c,
			\end{align*}
			thus contradicting the fact that our solution is $c$-approximate.
			
			The strategy profile $\sp'$ is obtained by modifying $\sp$ as follows. Blue agents that occupy dummy nodes exchange their position with red agents occupying nodes of the form $x_v$, with $v \in V(G')$. At the same time, every blue agent that occupies a node $y_e^i$, with $e \in E(G')$ and $i\in \{1,2\}$, exchanges its position with a red agent occupying a node $x_v$ such that $v \in V(G')$, where we give priority to nodes that cover $e$. Clearly, given $\sp$, strategy profile $\sp'$ can be computed in polynomial time and we have $\hn(\sp')\geq \hn(\sp)$. 
	
			Let $m''$ be the number of edges of $G'$ that are not covered by $$ C(\sp'):=\{v \in V(G') \mid x_v \textrm{ is occupied by a blue agent}\}. $$ We show that $m''\leq (c'-1)k^*$. First of all, we observe  all the $5m'$ dummy nodes, the special node $z$ and all the $n'$ nodes~$x_v$ corresponding to the nodes $v \in V(G')$ are not segregated in $\sp'$. Therefore, the number of uncovered edges of~$G'$ equals twice the number of segregated red agents in~$G$ (two red agents per uncovered edge $e$ of $G'$ that occupy the nodes $y_e^1$ and $y_e^2$). Therefore $\hn(\sp')=n-2m''$. As a consequence
			\begin{align*}
				c > \frac{\hn(\sp^*)}{\hn(\sp')} & \geq \frac{n}{n-2m''},
			\end{align*}
			from which we derive
			\begin{align*}
				m'' & \leq \frac{c-1}{2c}n = \frac{\frac{13}{14-c'}-1}{2\frac{13}{14-c'}}n = \frac{c'-1}{26}(n'+7m'+1) \\
				& \leq \frac{c'-1}{26}\left(\frac{2}{3}m'+7m'+m'\right) = \frac{c'-1}{3}m' \\
				& \leq (c'-1)k^*,
			\end{align*}
			where we use the facts that $m'=\frac{3}{2}n'$ (i.e., $G'$ is cubic) , $1 \leq m'$, and $k^* \geq \frac{m'}{3}$ (each node of $G'$ covers 3 edges).
	
			To complete the proof, let $C$ be a vertex cover of $G'$ that contains $C(\sp')$ and a node that covers each of the edges of~$G'$ that are not covered by $C(\sp')$. Clearly, given $\sp'$, $C$ can be computed in polynomial time. The size of $C$ is upper bounded by the size of $C(\sp')$ plus the number of uncovered edges, i.e., $$|C| \leq k^* + m'' \leq c' k^*. $$ Hence, $C$ is a $c'$-approximate vertex cover of $G'$. This completes the proof.
		\end{proof}

		\begin{theorem}
			There is a constant $c>1$ such that, given a game $(G,b,1/2)$ on a bipartite graph $G$, the problem of computing an SE $\sp$ that maximizes $\hn(\sp)$ is not $c$-approximable in polynomial time, unless $\mathsf{P}=\mathsf{NP}$.
		\end{theorem}
		\begin{proof}
			We consider exactly the same reduction that we used in the proof of Theorem \ref{thm:reduction_vertex_cover} and show  the existence of a strategy profile $\sp^*$ that maximizes $\hn(\sp^*)$ which is also an SE. Observe that, once we prove that $\sp^*$ is an SE, the rest of the proof can be derived from Theorem \ref{thm:reduction_vertex_cover}.
			
			Consider the strategy profile $\sp^*$ in which a blue agent occupies the special node $z$, while the remaining $k^*$ blue agents are placed on nodes of the form $x_v$ such that $v \in C^*$ and $C^*$ is an optimal vertex cover of $G'$ (see also Figure \ref{fig:vertexcover}). In the proof of Theorem \ref{thm:reduction_vertex_cover} we already showed that $\hn(\sp^*)=n$. In the following we show that $\sp^*$ is also a SE.
	
			The red agents on the dummy nodes have maximum utility, so they never swap.
			
			Let $j$ be a red agent that is placed on a node of the form~$x_v$, with $v \in V(G')$.  The strategy $\ssp$ where $i$ is a blue agent placed on a node $x_v$, with $v \in V(G')$, is not a profitable swap by Lemma \ref{lemma4}. The strategy $\ssp$ where $i$ is the blue agent placed on the special node $z$ is not a profitable swap by Lemma \ref{lemma0.1}.
	
			Finally, consider any red agent $j$ that is placed on a node of the form $y_e^\ell$ with $e \in E(G')$ and $\ell \in \{1,2\}$. This agent has a utility of either $p(1/3)$ or $p(2/3)$. But $p(1/3) = p(2/3)$ whenever $\Lambda \leq 1/2$. The strategy $\ssp$, where $i$ is a blue agent placed on a node $x_v$ with $v \in V(G')$, is not a profitable swap either by Lemma \ref{lemma0.1} (when $\sigma(i)$ is adjacent to $\sigma(j)$) or simply because the utility of $j$ in $\ssp$ is $p(7/8)<p(2/3)$. The blue agent $i$ on the special node $z$ has a utility that is strictly smaller than $p(1/6)$ as $$\frac{k^*+1}{5m'+n'+1} \leq \frac{n'}{5n'+n'+1}< 1/6.$$ Therefore, $\ssp$ is not a profitable swap because the utility of~$j$ in $\ssp$ is upper bounded by $p(5/6)<p(2/3)$.
		\end{proof}
	
	\section{Conclusion and Future Work}
		We study game-theoretic residential segregation with integration-oriented agents and thereby open up the novel research direction of considering non-monotone utility functions. Our results clearly show that moving from monotone to non-monotone utilities yields novel structural properties and different results in terms of equilibrium existence and quality. We have equilibrium existence for a larger class of graphs, compared to~\citep{AEGISV21}, and it is an important open problem to prove or disprove if swap equilibria for our model with $\Lambda\leq \frac12$ are guaranteed to exist on any graph. 
	
		So far we considered single-peaked utilities that are supported by data from real-world sociological polls. 
		However, also other natural types of non-monotone utilities could be studied. Also ties in the utility function could be resolved by breaking them consistently towards favoring being in the minority or being in the majority. The non-existence example of swap equilibria used in the proof of Theorem~\ref{thm:noSE} also applies to the case with $\Lambda = \frac12$ and breaking ties towards being in the majority. Interestingly, by breaking ties the other way we get the same existence results as without tie-breaking and also our other results hold in this case. This is another indication that tolerance helps with stability.
	
		Moreover, all our existence results also hold for utility functions having a symmetric plateau shape around~$\Lambda$. Investigating the PoA for these utility functions is open.
	
		Regarding the quality of the equilibria, we analyze the degree of integration as social welfare function, as this is in-line with considering integration-oriented agents. Of course, studying the quality of the equilibria in terms of the standard utilitarian social welfare, i.e., $\ut(\sp) = \sum_{i=1}^n \u_i(\sp)$, would also be interesting. We note in passing that on ring topologies the PoA and the PoS with respect to both social welfare functions coincide. 
	
\bibliographystyle{named}
\bibliography{singlepeak}

\begin{thebibliography}{}

\bibitem[\protect\citeauthoryear{Acevedo-Garcia and
  Lochner}{2003}]{acevedo2003residential}
D.~Acevedo-Garcia and K.~A. Lochner.
\newblock Residential segregation and health.
\newblock {\em Neighborhoods and Health}, pages 265--87, 2003.

\bibitem[\protect\citeauthoryear{Agarwal \bgroup \em et al.\egroup
  }{2021}]{AEGISV21}
A.~Agarwal, E.~Elkind, J.~Gan, A.~Igarashi, W.~Suksompong, and A.~A. Voudouris.
\newblock Schelling games on graphs.
\newblock {\em Artif. Intell.}, 301:103576, 2021.

\bibitem[\protect\citeauthoryear{Alimonti and Kann}{1997}]{alimonti97}
P.~Alimonti and V.~Kann.
\newblock Hardness of approximating problems on cubic graphs.
\newblock In {\em Algorithms and Complexity}, pages 288--298, 1997.

\bibitem[\protect\citeauthoryear{Barmpalias \bgroup \em et al.\egroup
  }{2014}]{BEL14}
G.~Barmpalias, R.~Elwes, and A.~Lewis-Pye.
\newblock Digital morphogenesis via schelling segregation.
\newblock In {\em FOCS 2014}, pages 156--165, 2014.

\bibitem[\protect\citeauthoryear{Betzler \bgroup \em et al.\egroup
  }{2013}]{BSU13}
N.~Betzler, A.~Slinko, and J.~Uhlmann.
\newblock On the computation of fully proportional representation.
\newblock {\em JAIR}, 47:475--519, 2013.

\bibitem[\protect\citeauthoryear{Bil{\`{o}} \bgroup \em et al.\egroup
  }{2020}]{BBLM20}
D.~Bil{\`{o}}, V.~Bil{\`{o}}, P.~Lenzner, and L.~Molitor.
\newblock Topological influence and locality in swap schelling games.
\newblock In {\em {MFCS 2020}}, pages 15:1--15:15, 2020.

\bibitem[\protect\citeauthoryear{Black}{1948}]{Black48}
D.~Black.
\newblock On the rationale of group decision-making.
\newblock {\em J. Pol. E.}, 56(1):23--34, 1948.

\bibitem[\protect\citeauthoryear{Boehmer and Elkind}{2020}]{BE20}
N.~Boehmer and E.~Elkind.
\newblock Individual-based stability in hedonic diversity games.
\newblock In {\em {AAAI 2020}}, pages 1822--1829, 2020.

\bibitem[\protect\citeauthoryear{Bogomolnaia and Jackson}{2002}]{BJ02}
A.~Bogomolnaia and M.~O. Jackson.
\newblock The stability of hedonic coalition structures.
\newblock {\em GEB}, 38(2):201--230, 2002.

\bibitem[\protect\citeauthoryear{Brandt \bgroup \em et al.\egroup
  }{2012}]{BIK12}
Ch. Brandt, N.~Immorlica, G.~Kamath, and R.~Kleinberg.
\newblock An analysis of one-dimensional schelling segregation.
\newblock In {\em STOC 2012}, pages 789--804, 2012.

\bibitem[\protect\citeauthoryear{Brandt \bgroup \em et al.\egroup
  }{2015}]{BBHH15}
F.~Brandt, M.~Brill, E.~Hemaspaandra, and L.~A. Hemaspaandra.
\newblock Bypassing combinatorial protections: Polynomial-time algorithms for
  single-peaked electorates.
\newblock {\em JAIR}, 53:439--496, 2015.

\bibitem[\protect\citeauthoryear{Bredereck \bgroup \em et al.\egroup
  }{2019}]{BEI19}
R.~Bredereck, E.~Elkind, and A.~Igarashi.
\newblock Hedonic diversity games.
\newblock In {\em {AAMAS 2019}}, pages 565--573, 2019.

\bibitem[\protect\citeauthoryear{Bullinger \bgroup \em et al.\egroup
  }{2021}]{BSV21}
M.~Bullinger, W.~Suksompong, and A.~A. Voudouris.
\newblock Welfare guarantees in schelling segregation.
\newblock {\em J. Artif. Intell. Res.}, 71:143--174, 2021.

\bibitem[\protect\citeauthoryear{Cable}{2013}]{C13}
D.~Cable.
\newblock The racial dot map.
\newblock {\em Weldon Cooper Center for Public Service, University of
  Virginia}, 2013.

\bibitem[\protect\citeauthoryear{Chan \bgroup \em et al.\egroup }{2020}]{CIT20}
H.~Chan, M.~T. Irfan, and C.~V. Than.
\newblock Schelling models with localized social influence: {A} game-theoretic
  framework.
\newblock In {\em {AAMAS 2020}}, pages 240--248, 2020.

\bibitem[\protect\citeauthoryear{Chauhan \bgroup \em et al.\egroup
  }{2018}]{CLM18}
A.~Chauhan, P.~Lenzner, and L.~Molitor.
\newblock Schelling segregation with strategic agents.
\newblock In {\em SAGT 2018}, pages 137--149, 2018.

\bibitem[\protect\citeauthoryear{Drèze and Greenberg}{1980}]{DG80}
J.~H. Drèze and J.~Greenberg.
\newblock Hedonic coalitions: Optimality and stability.
\newblock {\em Econometrica: J. Eco. Soc.}, pages 987--1003, 1980.

\bibitem[\protect\citeauthoryear{Echzell \bgroup \em et al.\egroup
  }{2019}]{E+19}
H.~Echzell, T.~Friedrich, P.~Lenzner, L.~Molitor, M.~Pappik, F.~Sch{\"{o}}ne,
  F.~Sommer, and D.~Stangl.
\newblock Convergence and hardness of strategic schelling segregation.
\newblock In {\em WINE 2019}, pages 156--170, 2019.

\bibitem[\protect\citeauthoryear{Elkind \bgroup \em et al.\egroup
  }{2014}]{EFS14}
E.~Elkind, P.~Faliszewski, and P.~Skowron.
\newblock A characterization of the single-peaked single-crossing domain.
\newblock In {\em {AAAI 2014}}, 2014.

\bibitem[\protect\citeauthoryear{Immorlica \bgroup \em et al.\egroup
  }{2017}]{BIK17}
N.~Immorlica, R.~Kleinberg, B.~Lucier, and M.~Zadomighaddam.
\newblock Exponential segregation in a two-dimensional schelling model with
  tolerant individuals.
\newblock In {\em SODA 2017}, pages 984--993, 2017.

\bibitem[\protect\citeauthoryear{Jackson \bgroup \em et al.\egroup
  }{2000}]{jackson2000relation}
S.~A. Jackson, R.~T. Anderson, N.~J. Johnson, and P.~D. Sorlie.
\newblock The relation of residential segregation to all-cause mortality: A
  study in black and white.
\newblock {\em Am. J. Pub. Health}, 90(4):615, 2000.

\bibitem[\protect\citeauthoryear{Kanellopoulos \bgroup \em et al.\egroup
  }{2021a}]{KKV21}
P.~Kanellopoulos, M.~Kyropoulou, and A.~A. Voudouris.
\newblock Modified schelling games.
\newblock {\em TCS}, 880:1--19, 2021.

\bibitem[\protect\citeauthoryear{Kanellopoulos \bgroup \em et al.\egroup
  }{2021b}]{KKV21a}
P.~Kanellopoulos, M.~Kyropoulou, and A.~A. Voudouris.
\newblock Not all strangers are the same: The impact of tolerance in schelling
  games.
\newblock {\em CoRR}, abs/2105.02699, 2021.

\bibitem[\protect\citeauthoryear{Kreisel \bgroup \em et al.\egroup
  }{2021}]{KBFN21}
L.~Kreisel, N.~Boehmer, V.~Froese, and R.~Niedermeier.
\newblock Equilibria in schelling games: Computational complexity and
  robustness.
\newblock {\em CoRR}, abs/2105.06561, 2021.

\bibitem[\protect\citeauthoryear{Massey and Denton}{1988}]{MD88}
D.~S. Massey and N.~A. Denton.
\newblock The dimensions of residential segregation.
\newblock {\em Soc. Forces}, 67(2):281--315, 1988.

\bibitem[\protect\citeauthoryear{Massey and Denton}{1993}]{massey1993american}
D.~S. Massey and N.~A. Denton.
\newblock {\em American Apartheid: Segregation and the Making of the
  Underclass}.
\newblock Harvard U. Press, 1993.

\bibitem[\protect\citeauthoryear{Monderer and Shapley}{1996}]{MS96}
D.~Monderer and L.~S. Shapley.
\newblock Potential games.
\newblock {\em GEB}, 14(1):124--143, 1996.

\bibitem[\protect\citeauthoryear{Schelling}{1971}]{Schelling71}
Th.~C. Schelling.
\newblock Dynamic models of segregation.
\newblock {\em J. Math. Soc.}, 1(2):143--186, 1971.

\bibitem[\protect\citeauthoryear{Schelling}{2006}]{S06}
Th.~C. Schelling.
\newblock {\em Micromotives and Macrobehavior}.
\newblock WW Norton \& Company, 2006.

\bibitem[\protect\citeauthoryear{Smith \bgroup \em et al.\egroup }{2019}]{gss}
Tom~W. Smith, M.~Davern, J.~Freese, and S.~L. Morgan.
\newblock General social surveys, 1972--2018 cumulative codebook.
\newblock {\em NORC ed. Chicago: NORC 2019, U. Chicago}, 2019.

\bibitem[\protect\citeauthoryear{Vazirani}{2013}]{vazirani}
V.~V. Vazirani.
\newblock {\em Approximation algorithms}.
\newblock Springer Science \& Business Media, 2013.

\bibitem[\protect\citeauthoryear{Walsh}{2007}]{Walsh07}
T.~Walsh.
\newblock Uncertainty in preference elicitation and aggregation.
\newblock In {\em AAAI 2007}, pages 3--8, 2007.

\bibitem[\protect\citeauthoryear{Yu \bgroup \em et al.\egroup }{2013}]{YCE13}
L.~Yu, H.~Chan, and E.~Elkind.
\newblock Multiwinner elections under preferences that are single-peaked on a
  tree.
\newblock In {\em {AAAI 2013}}, pages 425--431, 2013.

\bibitem[\protect\citeauthoryear{Zhang}{2004a}]{Zha04}
J.~Zhang.
\newblock A dynamic model of residential segregation.
\newblock {\em J. Mat. Soc.}, 28(3):147--170, 2004.

\bibitem[\protect\citeauthoryear{Zhang}{2004b}]{Zha04b}
J.~Zhang.
\newblock Residential segregation in an all-integrationist world.
\newblock {\em J. Econ. Behav. and Org.}, 54(4):533--550, 2004.

\end{thebibliography}

\end{document}